\documentclass[11pt, a4paper, twoside]{amsart}
\usepackage{hyperref}
\usepackage{amsmath,amssymb,latexsym,amsthm}
\usepackage{a4wide}
\usepackage{graphicx}
\usepackage{pstricks}
\usepackage{siunitx}
\usepackage{comment}
\usepackage{tabularx}
\usepackage{booktabs}

\newtheorem {lemma}{Lemma}

\newtheorem {theorem}{Theorem}

\newcommand{\R}{{\mathbb{R}}}
\renewcommand{\v}[1]{{\boldsymbol{#1}}} 
\newcommand{\m}[1]{{\boldsymbol{#1}}} 

\newcommand{\vb}{\v{b}}
\newcommand{\vc}{\v{c}}
\newcommand{\vd}{\v{d}}
\newcommand{\ve}{\v{e}}
\newcommand{\vf}{\v{f}}
\newcommand{\vm}{\v{m}}
\newcommand{\vphi}{\v{\phi}}
\newcommand{\vpsi}{\v{\psi}}
\newcommand{\vrho}{\v{\rho}}
\newcommand{\vr}{\v{r}}
\newcommand{\vs}{\v{s}}
\newcommand{\vt}{\v{t}}
\newcommand{\vu}{\v{u}}
\newcommand{\vtheta}{\v{\theta}}
\newcommand{\vx}{\v{x}}
\newcommand{\vxi}{\v{\xi}}
\newcommand{\vxp}{\v{x}_p}
\newcommand{\vy}{\v{y}}
\newcommand{\vz}{\v{z}}

\newcommand{\ma}{\m{A}}
\newcommand{\mb}{\m{B}}

\newcommand{\mh}{\m{H}}
\newcommand{\mi}{\m{I}}
\newcommand{\mgamma}{\m{\Gamma}}

\newcommand{\mk}{\m{K}}

\newcommand{\mm}{\m{M}}
\newcommand{\mq}{\m{Q}}
\newcommand{\psdn}{\Phi_n} 
\newcommand{\psdf}{\Phi_f} 
\newcommand{\psd}{\Phi} 
\newcommand{\psdg}{\Phi_{\text{gr}}} 
\newcommand{\psdr}{{\widetilde\Phi}} 
\newcommand{\psdk}{\Phi^{vK}} 
\newcommand{\psdgr}{{\widetilde\Phi}_{\text{gr}}} 
\newcommand{\vpsd}{\vphi} 
\newcommand{\cpsd}{\phi} 
\newcommand{\cor}{C}
\newcommand{\NA}{---}

\newcommand{\T}{{\text{T}}} 

\newcommand{\Ft}{{\mathcal F}}

\newcommand{\norm}[1]{\left\|#1 \right\|}
\newcommand{\abs}[1]{\left| #1\right|}

\def \E {{\mathbb{E}}}

\def\ap{D}

\DeclareMathOperator{\sinc}{sinc}
\DeclareMathOperator{\diag}{diag}

\newcommand{\var}{\alpha}

\title{Atmospheric turbulence profiling with unknown power spectral density}

\author{Tapio Helin}
\author{Stefan Kindermann}
\address{
Department of Mathematics and Statistics\\
University of Helsinki\\
Gustaf H\"allstr\"omin katu 2b\\
FI-00014 Helsinki\\
Finland\\}
\email{Tapio.Helin@helsinki.fi and Jonatan.Lehtonen@helsinki.fi}

\author{Jonatan Lehtonen}
\author{Ronny Ramlau}
\address{
Industrial Mathematics Institute \\
Johannes Kepler University \\
Altenbergerstra{\ss}e 69 \\
A-4040 Linz \\ Austria \\}
\email{Stefan.Kindermann@jku.at and Ronny.Ramlau@jku.at}

\date{\today}

\begin{document}

\begin{abstract}
Adaptive optics (AO) is a technology in modern ground-based optical telescopes 
to compensate for the wavefront distortions caused by atmospheric turbulence. One method that allows 
to retrieve information about the atmosphere from telescope data is 
so-called SLODAR, where the atmospheric turbulence profile is estimated based on correlation data of 
Shack--Hartmann wavefront measurements.
This approach relies on a layered Kolmogorov turbulence model.
In this article, we propose a novel extension of the SLODAR concept
by including a general non-Kolmogorov turbulence layer close to the ground with an unknown power spectral density. We prove that the joint estimation problem of the turbulence profile above ground simultaneously with the unknown power spectral density at the ground is ill-posed and propose three numerical reconstruction methods. 
We demonstrate by numerical simulations that our methods lead to  substantial improvements in the turbulence profile reconstruction compared to the standard SLODAR-type approach. Also, our methods can accurately locate local perturbations in non-Kolmogorov power spectral densities.
\end{abstract}
\maketitle

\section{Introduction}

Adaptive optics (AO) systems are designed to improve the imaging quality of ground-based optical telescopes by providing 
real-time compensation for the unwanted optical aberrations generated by atmospheric turbulence \cite{tyson2015}. Many next-generation 
AO systems aim to produce a diffraction-limited resolution in a large field of view. Based on observations of incoming light 
from several sources (guide stars), these wide-field AO systems estimate the turbulence volume (refractive index fluctuations) above the telescope before optimizing the optical correction by deformable mirrors. The crux of this challenge is a severely ill-posed inverse problem called
\emph{atmospheric tomography} \cite{ellerbroek2009}, where a three-dimensional scalar function describing the turbulence is estimated based on modified integral data over the atmosphere within milliseconds. 
Due to the extremely small angle of view (around 1--7 arcmin) and limited computational resources, any successful solution strategy in atmospheric tomography is based on reliable statistical modeling of the turbulence and effective discretization of the atmosphere. Our work aims at improving the tomographic 
reconstructions by optimizing such prior information empirically.

Atmospheric turbulence tends to be concentrated at the boundaries of different air flows and, consequently,
it can be described to a good approximation by a superposition of thin turbulent slabs, i.e., two-dimensional layers.
The tomographic reconstruction requires prior information on how to 
weight and locate these layers on the vertical axis. For an effective setup one needs to estimate the vertical \emph{turbulence profile}, i.e., the distribution of turbulent energy across different altitudes. Furthermore, there is an ongoing discussion and uncertainty on how to statistically model turbulence on these layers. In this paper we propose a method for optimizing the turbulence profile and statistical models simultaneously based on a time-series of data from a typical AO system. In particular, our novel contribution is to reconstruct the \emph{power spectral density} of the turbulence from an ill-posed integral equation.

In the past turbulence profiles have been traditionally measured using independent instruments but due to a number of advantages there is a need to estimate the profile within the AO system \cite{Gilles10}.
In fact, profiling based on the AO 
data can be performed by methods similar to SLOpe Detection And Ranging (SLODAR, see \cite{Wilson02}).
These methods utilize spatial correlations in the observations of incoming light (wavefront sensor measurements).
The first approach deduces the profile from the cross-correlations by deconvolution with the autocorrelation of the data \cite{Wilson02, Wang08}. The second approach \cite{guesalaga2017, Gilles10, cortes2012, Guesalaga2014, Butterley06,Vedrenne07} forms a linear dependency between the spatial cross-correlations and the vertical turbulence profile by assuming that the turbulence statistics at any altitude is accurately described by the Kolmogorov or von K\'arm\'an model \cite{panchev2016, kolmogorov1941, kolmogorov1962}.

This paper grows out of a number of experiments which confirm significant deviations of the turbulence statistics from the classical models in certain portions of the atmosphere
\cite{rao1999measuring, wu2014measurement, rao2000spatial, beland1995some}. 
In particular, a strong deviation would imply that the underlying assumption regarding the turbulence statistics in
the second SLODAR-based approach is violated and can lead to inaccurate profiling.
Deviations from classical models have been well-documented close to the ground \cite{Buser1971, Dayton1992, Bester1992} and in the upper troposhere and stratosphere \cite{kyrazis1994measurement, belen1999experimental, stribling1995optical}.
Moreover, the turbulence taking place in the telescope dome (so-called \emph{dome seeing}) is well-known to have a spectrum that deviates strongly from the Kolmogorov power law \cite{Guesalaga2014}. 
In consequence, there is a clear need to develop the current profiling methods further to take into account the uncertainty of the underlying turbulence statistics. 

Adaptive optics is strongly influenced by the turbulence at low altitudes since typically the turbulent layer with most energy is located close to the ground \cite{tyson2015}.
Our novel method estimates the turbulence profile \emph{and}
the full power spectral density of the so-called ground layer simultaneously based on the empirical time-averaged spatial correlations from the AO data similar to previous papers such as \cite{Gilles10, Butterley06}. In particular, the proposed method utilizes the data more effectively than previous methods.
Whereas the second approach described above leads to a linear overdetermined problem, 
the additional effort of estimating the ground layer turbulence makes the problem ill-posed. 
Thus, it is of vital importance to employ an appropriate regularization and \emph{a priori} information to obtain reasonable results.

We propose three different numerical methods 
and demonstrate their performance by simulations. 
The first method can be viewed as parametric estimation, where regularization is 
carried out by strict parametrization of the power spectral density function.
This approach aims to efficiently reconstruct a turbulence
power spectral density that is given by a non-Kolmogorov power law.
The other two approaches are non-parametric 
regularization schemes, one based on Tikhonov regularization and the other on total variation, which we call \emph{two-step methods}. 
A major challenge with the non-parametric approach is the polynomial decay of the power spectral density and how to scale the effect of regularization. Initially, these methods approximate the power spectral density with a fixed power law by executing the first method. Next, the obtained power law is used as a scaling factor in the regularization functional.
We demonstrate that the two-step methods lead to a marked improvement in the turbulence profile reconstruction, compared to the reconstructions obtained by assuming Kolmogorov turbulence in the ground layer. Also, our methods can accurately locate local perturbations in non-Kolmogorov power spectral densities. A standard example of such perturbation is the viscosity bump \cite{Hill1978} close to the upper inertial bound. 
A rough estimate of the spatial frequency regime visible to the current wavefront sensor technology is discussed based on the numerical experiments. All numerical simulations are carried out with the MOST simulation environment developed by the AAO team at JKU Linz, Austria.

Let us briefly review some related work. The idea of 
extracting additional information from the cross-correlation data
in SLODAR-type profiling methods is not new. 
In \cite{guesalaga2017} the altitude-dependent outer scale is estimated simultaneously with the turbulence profile based on von K{\'a}rman turbulence.  This paper is particularly interesting for our work due to the altitude-dependent behavior of the power spectral density.
Let us also mention that the SLODAR-type methods can be used to perform wind profiling by
utilizing time-delayed cross-correlations between
all possible combinations of wavefront measurements from
the available wavefront sensor data \cite{Wang08, Guesalaga2014}.
A method for parametric estimation of the power law exponent of a non-Kolmogorov turbulence structure function is proposed in \cite{rao1999measuring}.
However, this approach assumes similar power spectral density throughout the atmosphere and does not include profiling.
For a general perspective on profiling, see \cite{guesalaga2015, rodriguez2014}. 
We also point out a method proposed by the authors in \cite{helin2017}, where the turbulence profile is optimized as part of the atmospheric tomography problem.
Finally, let us note that in this paper we make the crucial assumption that the turbulence is a statistically isotropic and homogeneous random field. In practice, this may not be the case for the dome seeing. The physics and statistical law of turbulence in the telescope dome is still a matter of research \cite{lombardi2010}. 

This paper is organized as follows. In section \ref{sec:profiling} we briefly discuss the turbulence statistics and describe the starting point for this study: a SLODAR-type
atmospheric profiling method introduced in \cite{Butterley06} and extended to laser guide stars in \cite{Gilles10}. 
Section \ref{sec:inverse_problem} covers the identification of the ground layer and what properties this ill-posed problem has. We show that the related forward mapping is compact and establish certain approximation properties.
Numerical methods and their implementation are introduced in section \ref{sec:numerical_methods}, while section \ref{sec:simulations} discusses and demonstrates the performance of the methods in simulations.

\section{Atmospheric profiling}
\label{sec:profiling}

In this section we describe the SLODAR-type method for laser guide stars (LGS) based on the paper 
by Gilles and Ellerbroek \cite{Gilles10}. We begin with an introduction to the method and 
derive  the mathematical formulas that lead to a matrix equation for the atmospheric turbulence profile.

Let us briefly define the relevant notation. A random field $f : \R^n \to L^2(\Omega; \R)$ on a complete
probability space $(\Omega, \mathbb{P})$ is called \emph{stationary}, if its mean function $\mathbb{E} f$ is constant
and the autocorrelation function
$C_f({\bf x},{\bf y}) = \mathbb{E} (f({\bf x})  f({\bf y}))$ depends only on the difference ${\bf x}-{\bf y}$.
Moreover, $f$ is called \emph{isotropic} if $C_f$ depends only on the distance $|{\bf x}-{\bf y}|$. 
A structure function $D_f$ is defined according to
\begin{equation*}
	D_f(r) = \mathbb{E}(f({\bf x}+{\bf r})-f({\bf x}))^2,
\end{equation*}
where $r = |{\bf r}|$. For stationary and isotropic $f$, $D_f$ does not depend on $\vx$ or the direction of $\vr$.
Below, we mainly consider the \emph{power spectral density} (PSD) representation. Let us define the Fourier transform $\mathcal F \phi$ of a function $\phi$ through
\[ (\mathcal F \phi) (\vxi) = \int_{\R^n} e^{-2 \pi i \vx \cdot \vxi} \phi(\vx) d\vxi. \]
Then, by the Wiener--Khintchine theorem, the autocorrelation function $\cor_f$ of the random field $f$ coincides with the inverse Fourier transform of the PSD $\psdf$ \cite{Conan2008}:
\begin{equation*}
\cor_f(\vr) = (\mathcal F \psdf)^{-1} (\vr) = \int_{\R^n} e^{2 \pi i \vr \cdot \vxi} \psdf(\vxi) d\vxi.
\end{equation*}
In addition, we will denote the correlation of two random fields $f$ and $g$ by $C_{f,g}({\bf x},{\bf y}) = \mathbb{E}(f(\vx) g(\vy))$.

Above and throughout the rest of this article we use bold symbols to denote vectors, matrices, and vector-valued functions.

\subsection{Atmospheric turbulence}
\label{sec:turbulence}

Atmospheric turbulence is the highly irregular mixing and dynamics of the air in the atmosphere.
In this regard, atmospheric optics is concerned with how the refractive index of air behaves under turbulent flow, since the perturbations in the phase of passing light are proportional to the refractive index fluctuations. 
The spatial statistics of turbulence can be analyzed by the energy cascade principle and scaling arguments;
the classical work towards this goes back to Kolmogorov \cite{kolmogorov1941}, who was first to develop a model for turbulence structure functions. 

The refractive index of air is a function of temperature and humidity.
However, at astronomical sites, temperature is by far the dominating factor. 
Therefore, the statistics of refractive index fluctuations can be identified based on
the subsequent work by Tatarskii \cite{tatarskii}, who derived the structure function of the temperature fluctuations.
The key approximation relevant to atmospheric optics is that these fluctuations can be represented by an isotropic and stationary Gaussian random field.

The Kolmogorov model asserts that the structure function increases according to a $2/3$-power law within the inertial range $(\ell_0, L_0)$. Here, $\ell_0$ and $L_0$ are the characteristic sizes of the smallest and largest turbulent eddies in the atmosphere, respectively. The model is not realistic outside the inertial range and several adjustments have since been proposed in literature. Von K{\'a}rm{\'a}n proposed a model taking into account the so-called saturation regime of large distances.
The corresponding power spectral density of the refractive index fluctuations $n=n({\bf x})$ is given by
\begin{equation}
\label{eq:vonkarman}
	\psdn(\vxi, h) = a \rho(h)(\abs\vxi^2 + L_0^{-2})^{-11/6},
\end{equation}
where $\vxi$ is the spatial frequency and $a \approx 9.7 \times 10^{-3}$. An accurate description of the constant is given in \cite{Conan2008}.
The index of refraction structure constant $\rho = \rho(h)$ (also called the $C_n^2$-profile) is a function of altitude and is the object of interest for profiling methods.

As mentioned above, there is now available a significant body of evidence that conflicts with the conventional statistical model of turbulence based on various measurements of turbulence under a variety of conditions. The authors of \cite{toselli2008}
consider generalized models with varying power law in \eqref{eq:vonkarman} such that
\begin{equation}
	\label{eq:gen_vonkarman}
	\psdn(\vxi, h, \gamma) = a(\gamma) \rho(h)(\abs\vxi^2 + L_0^{-2})^{-\gamma},
\end{equation}
where $a(\gamma)$ is a constant normalizing the energy and $3 < \gamma < 4$.

In the following, our main object of interest, the power spectral density (PSD) for the phase fluctuations of light waves passing through a turbulent atmosphere is denoted by $\psd$. It is closely related to $\psdn$ physics, differing by a multiplicative constant. 
Notice carefully that later we consider general PSDs beyond power laws such as in \eqref{eq:vonkarman} or \eqref{eq:gen_vonkarman}. To this end, we have to take extreme care how the relative strength $\rho(h)$ is interpreted in our setting. 
We fix the following convention related to notations:
In the general discussion, we write $\psd = \psd(\vxi, h)$
to distinguish the altitude dependency of the PSD when relevant to context.
When the statistics is governed by the von K{\'a}rm{\'a}n model
we write $\psd(\vxi, h) = \rho(h) \psdk(\xi)$. Here, 
 $\psdk$ stands for the von K{\'a}rm{\'a}n PSD 
\begin{equation}
\label{eq:vonkarmanphase}
\psdk(\vxi) = b (\abs\vxi^2 + L_0^{-2})^{-11/6},
\end{equation}
when integrated over the whole atmosphere. Above, $b$ is a constant depending on some atmospheric parameters; we refer the reader to \cite{Conan2008} for details. 
Later, when we study the statistics of the ground layer, we write
$\psd(\vxi, 0) = \psdg(\vxi)$ and embed any constants to the PSD itself.

Finally, similar to \eqref{eq:gen_vonkarman}, we can also define the generalized version of the von K\'arm\'an PSD \eqref{eq:vonkarmanphase} to obtain
\begin{equation}
\label{eq:gen_vonkarmanphase}
\psdk(\vxi, \gamma) = b(\gamma) (\abs\vxi^2 + L_0^{-2})^{-\gamma}.
\end{equation}

\subsection{The setup of SLODAR}
\label{sec:slodar}

The goal of SLODAR-type methods is to determine the turbulence strength at each layer, i.e.,
the \emph{vertical turbulence profile}, by studying the correlation of measurements from two guide stars.
In this paper we consider \emph{laser 
guide stars} (LGS), which are artificially created by firing powerful lasers which scatter in the upper atmosphere 
and effectively create a star at a finite altitude $H$. This means we can choose the location of the guide stars, 
with some few additional adaptions (e.g., the \emph{cone compression}). More details on this can be found in \cite{Gilles10}.

\begin{figure}[ht]
\graphicspath{{images/}}
\centering{
\def\svgwidth{0.4\linewidth}
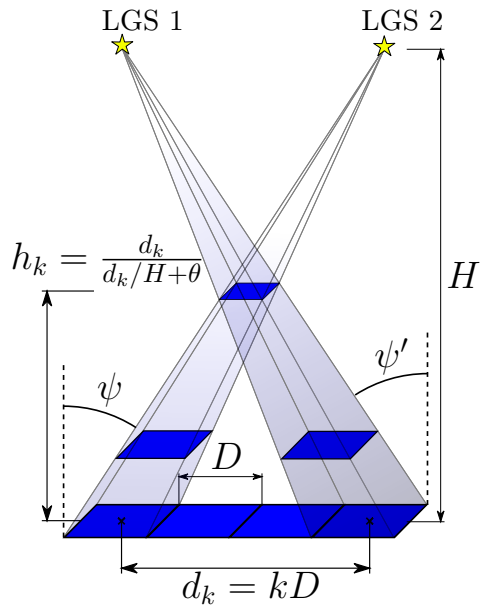
\caption{Illustration of the SLODAR measurement setup, showing a 
row of WFS subapertures measuring wavefronts from two LGSs. The correlation of these
measurements is a sum of the wavefront correlations between two square-shaped 
areas in each layer of turbulence. In particular, these areas coincide in the layer at the altitude $h_k$, 
which depends on subaperture separation $d_k$, LGS altitude $H$ and LGS separation $\theta = \psi - \psi'$.}
\label{fig:slodar}
}
\end{figure}

The wavefront aberrations are often 
observed using so-called Shack--Hartmann wavefront sensors (WFS). These sensors consist of a grid of square subapertures, and each subaperture effectively measures the average slope of incoming wavefronts. 
The WFSs coincide optically with the telescope lens and consequently observe an area as wide as the telescope lens. Moreover, each WFS observes a different direction. This measurement setup is illustrated by Fig.~\ref{fig:slodar}.

Suppose that subapertures have diameter $D$ and
consider the incoming wavefronts observed at subapertures that are optically separated by a distance $d_k = k D$, where $k$ is some integer.
Thus, the wavefronts measured by these subapertures both pass through the same volume of air at altitude $h_k = d_k / (d_k/H + \theta)$, where $H$ is the LGS altitude and $\theta = \psi - \psi'$ the LGS angle of separation.
Similarly, at altitude $h_j$ the separation of these wavefronts is $(d_k - d_j) \eta(h)$, where $\eta(h) = 1 - \tfrac{h}{H}$ is known as the \emph{cone compression factor}. In consequence, the cross-correlation in phase fluctuations of incoming wavefronts at the WFSs is the sum of such phase fluctuation cross-correlations through the turbulence layers.

The SLODAR method reconstructs the turbulence strengths by 
assuming that the turbulence statistics at any altitude is accurately described by the statistical turbulence models described in the previous section. This yields a
connection between the cross-correlations of WFS measurements and the turbulence profile, which we describe in detail in the following section. Measuring these correlations for a large number of separations $k = 0, \ldots, N$ determines the turbulence strength at certain altitudes $h_0, \ldots, h_N$.

\subsection{Model for WFS measurements}
\label{sec:wfs_measurements}
In this section we briefly describe an approach to atmospheric turbulence profiling discussed and analyzed in \cite{Gilles10}.
Our starting point is to write the Shack--Hartmann (SH) wavefront sensor measurement 
in the form
\begin{equation*}
\vs(\vxp) = \frac{1}{\ap^2} \int_0^\infty \int_{\R^2} \chi\left(\tfrac{\vx-\vxp}{\ap}\right) (\nabla \varphi)(\eta(h) \vx + h \vpsi, h) \eta(h) d\vx dh\in \R^2,
\end{equation*}
where $\vs(\vxp)$ stands for the SH-measurement at subaperture location $\vxp$ with 
a subaperture size $\ap$. Moreover, $\chi$ is the characteristic function of
the square $\left[-\frac 12,\frac 12\right]^2\subset \R^2$ and $\varphi :\R^2 \times [0,\infty) \to \R$ is 
the layer turbulence density. Above, the guide star is located at direction $\vpsi$ and 
the term $\eta(h) = 1- \tfrac{h}{H}$ is the cone compression factor mentioned in 
section \ref{sec:slodar}, for an LGS at an altitude $H>0$.

The global tip/tilt and focus are poorly measured by LGS WFSs, and therefore the noise in these modes propagates to cross-correlations. 
In \cite{Gilles10}, the authors avoid this unwanted effect by considering instead the following 
local curvature (second-order differences) of the measurements.
Let us write 
\[ m_x(\vxp) = s_x(\vxp - D \ve_x) - 2 s_x(\vxp) + s_x (\vxp + D \ve_x), \]
where $s_x$ is the $x$-coordinate of $\vs = (s_x,s_y)$ and $\ve_x = (1,0)$ is the horizontal unit vector. 
The curvature $m_y(\vxp)$ with $s_y$ and $\ve_y = (0,1)$ is defined analogously. 
For notational simplicity, we 
adopt the common notation $m_\var(\vxp)$ with 
$\var \in \{x, y\}$ to denote $m_x$ and $m_y$, respectively.
 Defining the function $g_h^\var: \R^2 \to \R$ as 
 \begin{equation*}
 	g_h^\var(\vt) = 8 \pi i \eta(h) t_\var \ap^2 \sin^2(\pi \eta(h) \ap t_\var) \sinc(\eta(h) \ap t_1) 
 	\sinc (\eta(h) \ap t_2),  \qquad \var \in \{x, y\}
 \end{equation*}
and following the steps in \cite{Gilles10}, one can 
derive the following expression for $m_\var(\vxp)$:
\begin{equation*}
m_\var(\vxp) =  \frac{1}{\ap^2} \int_0^\infty \int_{\R^2} e^{2\pi i \vxi \cdot (\eta(h) \vxp +  h \vpsi)} g_h^\var(\vxi) 
	(\Ft \varphi)(\vxi,h) d\vxi dh,  \qquad \var \in \{x, y\}.
\end{equation*}

Consider next a pair of LGSs located at directions $\vpsi$ and $\vpsi'$, and denote
the corresponding measurements by $m_\var$ and $m_\var'$, respectively. 
Since the random field $\varphi$ is stationary, we find that the correlation 
\begin{equation*}
	\cor_{m_\var, m_\var'}(\vd) := \E \left[m_\var(\vxp) \overline{m_\var'(\vxp+\vd)}\right]
\end{equation*}
is independent of the spatial location $\vxp$. 

Now suppose that the power spectral density of the phase fluctuations is given by $\psd(\vxi,h)$.
Since layers at different altitudes are statistically independent, we have the formal distributional identity
\begin{equation*}
	\E \left[(\Ft \varphi)(\vxi,h)(\Ft \varphi)(\vxi',h')\right] = \delta(\vxi-\vxi') \delta(h-h') \psd(\vxi,h),
\end{equation*}
where $\delta$ stands for the Dirac delta. We conclude after some calculation using standard Fourier 
analysis that
\begin{equation}\label{eq:cor_cont}
\cor_{m_\var,m_\var'}(\ap\vd) = \frac{1}{\ap^4}
\int_0^\infty \int_{\R^2} 
e^{-2 \pi i \vxi \cdot \eta(h) \ap \vd} e^{2 \pi i \vxi \cdot h \vtheta}   
|g_h^\var(\vxi)|^2 \psd(\vxi,h) d\vxi dh,
\end{equation}
where we have denoted $\vtheta = \vpsi - \vpsi'$.  
In the following we will take $\vd$ to be an integer vector, so that $\ap\vd$ gives the distance between two subapertures. 

Let us next consider discretization of equation \eqref{eq:cor_cont}.
To this end, the first step is to adopt a discrete layer model by expressing the integral over $h$ as a finite sum over $N_L$ atmospheric layers:
\begin{equation}\label{eq:cor_discrete}
\cor_{m_\var,m_\var'}(\ap\vd) = \frac{1}{\ap^4}
\sum_{k=0}^{N_L-1} \int_{\R^2} 
e^{-2 \pi i \vxi \cdot \eta_k \ap \vd} e^{2 \pi i \vxi \cdot h_k \vtheta}   
|g_k^\var(\vxi)|^2 \psd(\vxi,h_k) d\vxi,
\end{equation}
where $h_k$ are discrete altitudes for $k = 0, \ldots, N_L - 1$, and we introduce 
the notation $\eta_k := \eta(h_k)$ and $g_k^\var(\vxi) := g_{h_k}^\var(\vxi).$ 

\subsection{Profiling based on correlation data}

The SLODAR-type method introduced in \cite{Gilles10} is based on the assumption of the von K\'{a}rm\'{a}n model in \eqref{eq:vonkarmanphase} at each altitude.
Given $N_d$ vectors $\vd_j$ and the discretized turbulence profile $\{\rho_k\}_{k=0}^{N_L - 1}$ for $\rho_k := \rho(h_k)$, the equation \eqref{eq:cor_discrete} reduces to problem
\begin{equation}
\label{eq:cor_aux}
\cor_{m_\var,m_\var'}(\ap\vd_j) =: \sum_{k=0}^{N_L-1} A_{jk}^\var \rho_k,
\end{equation}
where $A_{jk}^\var$ is given by
\begin{equation}
\label{eq:submatrix_A}
A_{jk}^\var = \frac{1}{\ap^4}
\left(\int_{\R^2} 
e^{-2 \pi i \vxi \cdot \eta_k \ap \vd_j} e^{2 \pi i \vxi \cdot h_k \vtheta}   
|g_k^\var(\vxi)|^2 \psdk(\vxi) d\vxi\right), \;\;\;\; j = 0, \ldots, N_d - 1.
\end{equation}
Given the matrix $\ma^\var = (A_{jk}^\var)_{j,k}$, it remains to choose the values for $\vd_j$ and $h_k$. The natural choice is to take $N_d = N_L$ and choose the separations $\vd_k$ and altitudes $h_k$ such that rays drawn from the points $\vxp$ and $\vxp + \ap \vd_k$ in directions $\vpsi$ and $\vpsi'$, respectively, intersect at the altitude $h_k \geq 0$, as in Fig.~\ref{fig:slodar}; this is the setup which was used by \cite{Gilles10}. Observe that an immediate consequence of this setup is that $\vtheta = \vpsi - \vpsi'$ and $\ap \vd_k$ must be parallel, and the altitude $h_k$ may be computed as
\[ h_k = \frac{\abs{\vd_k} \ap}{\abs{\vd_k} \ap / H + \abs{\vtheta}}. \]

A further restriction is needed to apply this method in practice. Namely, there should exist a subaperture with midpoint $\vxp$ in the WFS corresponding to the LGS in direction $\vpsi$, such that $\vxp + \ap\vd_k$ is the midpoint of a subaperture in the other WFS; this is required to be able to measure the correlation $\cor_{m_\var,m_\var'}(D \vd_k)$. Given $\vd_k$, we denote by $N_k$ the number of subaperture pairs which satisfy this condition, and call such pairs \emph{valid}. 

It should be noted that since $m_\var$ measures wavefront curvature and is defined as a weighted average over SH measurements of neighboring subapertures, there will be no measurements available for subapertures at the edge of the WFS. In addition, observe that the requirement for valid pairs implicitly restricts the directions $\vpsi$ and $\vpsi'$ of the LGSs and the location of the WFS subapertures.

In what follows we make the crucial assumption that the correlation 
in equation \eqref{eq:cor_aux} can be determined to a good approximation by taking empirical cross-correlations from a time series of WFS measurements. In the most simplistic setup, one performs several observations during a time interval such that time series averages reach ergodic limits. The length of the required time interval depends on several parameters including wind speed and the WFS setup. More discussion of this approximation is given at the end of section \ref{sec:simdesc}.

By abbreviating $b_j^\var = \cor_{m_\var,m_\var'}(\ap\vd_j)$ and $\vb^\var = (b_j^\var)_{j=1}^{N_d}$ we arrive at
\[ \vb^\var = \ma^\var \vrho, \]
where $\ma^\var$ is an $(N_d \times N_L)$--matrix, and $\vb^\var$ and $\vrho$ are vectors of lengths $N_d$ and $N_L$, respectively.
Further, since $\var \in \{x, y\}$, we obtain two matrix equations for $\vrho$, and combining them yields
\[ \vb := \begin{pmatrix} \vb^x \\ \vb^y \end{pmatrix} = \begin{pmatrix} \ma^x \\ \ma^y \end{pmatrix} \vrho =: \ma \vrho. \]
Since $\ma$ is now a $2 N_d \times N_L$--matrix, this problem is overdetermined when $N_d > N_L / 2$;
in particular, this is the case when $N_d = N_L$. In \cite{Gilles10}, the matrix equation is 
solved in the least-squares sense by defining the solution $\widetilde{\vrho}= \ma^\dagger \vb,$
where $\ma^\dagger$  is the Moore--Penrose pseudoinverse of $\ma$. 
The drawback of this approach is that it does not enforce the physical property
that $\vrho$ should be non-negative. In \cite{Gilles10}, Gilles and Ellerbroek solved this
problem by setting negative entries of $\widetilde{\vrho}$ to zero and removing the 
corresponding columns of $\ma$. The solution process was then repeated until a non-negative solution was obtained.

Another option for solving this problem is to include the non-negativity constraint in the 
minimization problem directly by setting
\begin{equation}\label{eq:standardslodar} \widetilde{\vrho} = \arg\min_{\vrho \geq 0} \norm{\ma \vrho - \vb}_2^2. \end{equation}
In this case the solution is not directly available via the pseudoinverse, but the problem
is still easy to solve using, for example, the \texttt{quadprog}-function in MATLAB. The method based on solving \eqref{eq:standardslodar} will from now on be referred to as the standard SLODAR method, to distinguish it from the methods which will be introduced later in Section~\ref{sec:numerical_methods}.

\subsection{Identification of the ground layer statistics}
\label{sec:inverse_problem}

As we pointed out in section \ref{sec:turbulence},
having a fixed power spectral density for each layer can be unreliable.
Hence, to improve the modeling we consider the power spectral density of the lowest turbulence layer to be unknown, and denote it by $\psdg(\vxi)$ to distinguish it from the known power spectral density at other altitudes.
Consequently, the revisited discrete problem can be written as
\begin{equation}\label{eq:revisited}
b_j^\var = (K_0^\var \psdg)(\vd_j) + \sum_{k=1}^{N_L-1} A_{jk}^\var \rho_k, \;\;\;\; j = 0, \ldots, N_d - 1,
\end{equation}
where $A_{jk}^\var$ is given by \eqref{eq:submatrix_A} and $K_0^\var$ is a certain integral operator detailed below.
Notice that $A_{j0}^\var$ can no longer be computed explicitly since $\psdg(\vxi)$ is an unknown. 
Our problem is then to
\begin{equation}
	\label{eq:problem_statement}
	\textrm{reconstruct } \psdg \textrm{ and } (\rho_k)_{k=1}^{N_L-1} \textrm{ given the vector } \vb.
\end{equation}

Let us next focus on the analytical  properties of the forward operator given in \eqref{eq:cor_cont}.
In particular, for both analytical and numerical purposes, it is important to understand the spaces
between which the integral operator maps. 
Moreover,  we show that the Fourier transform with kernel $\abs{g_k^\var(\vxi)}^2$ gives rise to a compact 
 integral operator, implying in particular that the problem is ill-posed. We will also provide an 
 estimate on the decay of the integral, to justify approximating the infinite integral in \eqref{eq:cor_cont} by an 
 integral over a finite domain in the numerical implementation of the algorithms; cf.~Theorem~\ref{thm:compact}.

 It is reasonable  to assume that $\psdg \in L^2(\R^2)$, and as we deal with radially symmetric 
power spectral densities, we define  the $L^2$-space of such functions as  
\[ L_\text{rad}^2(\R^2)  = \left\{ f(x)\colon \R^2\to \R \;|\; f(\vxi) = \widetilde f(\abs\vxi), 
\int_0^\infty \abs{\widetilde f(r)}^2 r dr < \infty \right\}. \]
Obviously,  this is the space of all radially symmetric $L^2(\R^2)$-functions. 
In the following,  a tilde,  $\widetilde \Phi$, denotes the radial part of a radially symmetric function in $L^2(\R^2)$.

Now, define the integral operator $K_h^\var \colon L_\text{rad}^2(\R^2) \to L^2(\R^2)$ 
\begin{equation}\label{eq:int_op}
(K_h^\var \Phi)(\vd) := \int_{\R^2} e^{-2 \pi i \vxi \cdot \vd} |g_h^\var(\vxi)|^2 \Phi(\vxi) d\vxi.
\end{equation}
This corresponds to  the operator we used earlier in equation \eqref{eq:revisited}, where vectors $\vd_j$ are as in the 
previous 
section;  setting $h=0$ yields the forward operator for the ground layer, $K_0^\var$, as above. 
Nevertheless, for the following results consider a general altitude $h$ in  $K_h^\var$.
We also consider a cut-off version of $K_h^\var$ as follows: 
\[ (K_{h,n}^\var \psd)(\vd) = \int_{B(0,n)} e^{-2 \pi i \vxi \cdot \vd} \abs{g_h^\var(\vxi)}^2 \psd(\vxi) d\vxi. \]

The next theorem states on the one hand  that $K_h^\var$ is a compact operator and hence that the ground-layer 
identification is ill-posed, such that  regularization has to be applied. On the other hand, we also provide  an error
estimate for the cut-off approximation in order to justify the numerical approximation of the integral that we 
use below. 
\begin{theorem}\label{thm:compact}
The integral operator $K_h^\var \colon L_\text{rad}^2(\R^2) \to L^2(\R^2)$ defined in \eqref{eq:int_op} is
compact for all $0 \leq h < H$ and $\var \in \{x, y\}$.

Let $\psd \in L_\text{rad}^2(\R^2)$ be non-negative. Then for all $0 \leq h < H$ and $\var \in \{x,y\}$,
we have
\[ \norm{(K_h^\var  - K_{h,R}^\var)\psd}_\infty \leq 64 D^2 C(h) \int_R^\infty \psdr(r) dr, \]
for all $R \geq 1$, where $\psdr(\abs\vxi) = \psd(\vxi)$.
\end{theorem}

\begin{proof}
Since the kernel of the integral operator is continous, if follows from classical results that 
$K_{h,n}^\var$ is a compact operator from  $L_\text{rad}^2(\R^2) \to L^2(\R^2)$.  
Fix $\psd$ such that $\norm{\psd}_{L_\text{rad}^2(\R^2)} = 1$.
Since the 
 Fourier transform is unitary on $L^2$ and by 
Lemma \ref{lemma:g} (see Appendix~\ref{sec:appendix}),  we find 
\begin{align*}
\norm{(K_h^\var  - K_{h,n}^\var)\psd}^2 \leq 
&
\int_{\R^2 \setminus B(0,n)} \abs{g_h^\var(\vxi)}^4 \psd(\vxi)^2 d\vxi \leq 
 \norm{g_h^\var}_\infty^2 \int_{\R^2 \setminus B(0,n)}   \abs{g_h^\var(\vxi)}^2 \abs{\psd(\vxi)}^2 d\vxi \\
& \leq   \norm{g_h^\var}_\infty^2 C(h)  \int_n^\infty \abs{\psdr(r)}^2 dr
\leq   \frac{1}{n} \norm{g_h^\var}_\infty^2 C(h)   \norm{\psdr}_{L_\text{rad}^2(\R^2)}^2,
\end{align*}
where we have used that $\frac{r}{n} \geq 1$. 
It can be elementarily verified that 
$\norm{g_h^\var}_\infty \leq 8 D$, and thus the estimate in the theorem is valid. Moreoever, 
taking the limit $n\to \infty$ yields compactness of $K_h^\var$ as the uniform limit of 
a compact operator is again compact \cite{Kress14}.
\end{proof}

We remark  that the generalized von K\'{a}rm\'{a}n model $\psdk$ in
equation \eqref{eq:gen_vonkarmanphase} satisfies the assumptions in Theorem~\ref{thm:compact}
when $b(\gamma), L_0 > 0$ and 
$\gamma > \tfrac 12$. In this case we obtain the estimate
\[ \norm{(K_h^\var   - K_{h,R}^\var )\psdk}_\infty \leq\frac{C(h) b(\gamma)}{1 - 2 \gamma} R^{1-2 \gamma}. \]

The compactness of  $K_h^\var$ in Theorem~\ref{thm:compact} implies the ill-posedness of the ground layer identification problem, 
which is, of course, not unexpected. In the next section we propose various regularization penalties 
or priors that stabilize the problem and that are adapted to the most important case of power-law spectral 
densities.

\section{Three methods of regularization and their implementation}
\label{sec:numerical_methods}
Let us now discuss different approaches to effectively regularize the problem. The first method can be seen as regularizing the problem by strict parametrization. The other two methods, however,
are non-parametric as we allow for general power spectral densities. 
\subsection{Discretizing the PSD}
We begin by assuming that $\psdg$ is non-negative and radially symmetric, and express it as a linear combination of $N_R$ different radial basis functions centered at the origin. Recall that $\rho_0$ is embedded to the definition of $\psdg$ and write
\begin{equation*}
\psdg(\vxi) = \sum_{l=0}^{N_R - 1} \cpsd_l f_l(\abs \vxi),
\end{equation*}
where $\cpsd_l$ are non-negative coefficients and $f_l(r): [0,\infty) \to [0, \infty)$ are non-negative radial basis functions satisfying
\[ f_l(r_j) = \delta_{jl} = \begin{cases} 1, & j = l, \\ 0, & \text{otherwise} \end{cases} \]
for given discretization points $0 \leq r_0 < r_1 < \cdots < r_{N_R - 1}$. 
Below we used trigonometric basis functions, to enforce some additional smoothness on the solutions. These functions were chosen so that $f_l$ is supported in $[r_{l-1}, r_{l+1}]$ and the functions $f_l$ form a partition of unity in the interval $[0, r_{N_R-1}]$.

With these definitions, we can approximate
\begin{equation*}
(K_0^\var \Phi)(\vd_j) = \frac{1}{\ap^4}
\sum_{l=0}^{N_R - 1} \left(\int_{\R^2} 
e^{-2 \pi i \vxi \cdot \eta_0 \ap \vd_j} e^{2 \pi i \vxi \cdot h_0 \vtheta}   
|g_k^\var(\vxi)|^2 f_l(\abs{\vxi}) d\vxi\right) \cpsd_l = \sum_{l=0}^{N_R-1} B_{jl}^\var \cpsd_l, 
\end{equation*}
for $j = 0, \ldots, N_d - 1$, where the matrix element $B_{jl}^\var$ is given by
\[ B_{jl}^\var := \frac{1}{\ap^4} \left(\int_{\R^2} 
e^{-2 \pi i \vxi \cdot \eta_0 \ap \vd_j} e^{2 \pi i \vxi \cdot h_0 \vtheta}   
|g_k^\var(\vxi)|^2 f_l(\abs{\vxi}) d\vxi\right). \]
We then obtain the matrix equation
\[ \vb^\var = \widetilde{\ma}^\var \widetilde{\vrho} + \mb^\var \vpsd, \]
where $\widetilde{\ma}^\var$ denotes the submatrix of $\ma^\var$ with the first column removed, and similarly $\widetilde{\vrho}$ denotes $\vrho$ with the first element removed. Combining the cases $\var = x$ and $\var = y$ finally yields
\begin{equation}\label{eq:final}
\vb = \begin{pmatrix} \vb^x \\ \vb^y\end{pmatrix} = \begin{pmatrix}\widetilde \ma^x \\ \widetilde \ma^y\end{pmatrix} \widetilde \vrho + \begin{pmatrix}\mb^x \\ \mb^y\end{pmatrix} \vpsd =: \ma \vrho + \mb \vpsd, 
\end{equation}
which we wish to solve for $\vrho$ and $\vpsd$. Note that we have abused notation slightly and dropped the tildes from $\ma$ and $\vrho$, since the first column of $\ma$ and the first entry of $\vrho$ are now represented in the term $\mb \vpsd$.

Observe that since $b^x$ and $b^y$ are vectors of length $N_d$, we have a total of $2 N_d$ measurements and $N_L + N_R - 1$ unknowns. Taking $N_R \gg N_d$ leads to a strongly under-determined system, and thus solving this system for $\vrho$ and $\vpsd$ is an ill-posed problem. 

\subsection{Method 1: Parametric Power Law}

Our first method is based on the assumption that the ground layer statistics satisfy a generalized von K\'{a}rm\'{a}n model similar to \eqref{eq:gen_vonkarmanphase}:
\begin{equation*}
	\psdg(\vxi) = c(\gamma) (|\vxi|^2 + L_0^{-2})^{-\gamma}
\end{equation*}
for some constant $c(\gamma)>0$.
Employing a  least-squares approach, we thus consider the minimization problem
\begin{equation}\label{eq:powerfit} \min_{\vrho, c, \gamma \geq 0} \norm{\ma \vrho + \mb \vpsd(c, \gamma) - \vb}_2, \end{equation}
where $\vpsd(c, \gamma)$ gives a vector where the $j$'th element is the value of the power 
law at the discretization point $r_j$, defined as
$$\cpsd_j(c, \gamma) = c (r_j^2 + 1/L_0^2)^{-\gamma}. $$
Note that the nonlinear dependence of $\vpsd(c, \gamma)$ on the parameters $c$ and $\gamma$ makes the whole problem nonlinear. Since $\vrho$ is now a vector of $N_L - 1$ elements, 
the minimization problem \eqref{eq:powerfit} has a total of $N_L + 1$ unknowns, and they are all subject to non-negativity constraints. We solved the problem using \texttt{fmincon}, the Matlab function for constrained nonlinear multivariate minimization, which works quite quickly with so few unknowns. We used the \texttt{'sqp'}-algorithm with default options except for an optimality tolerance of $10^{-8}$ and a limit of 30000 iterations. For all three methods, we normalized the measurement vector $\vb$ to avoid numerical problems.

Note that we do not take the outer scale $L_0$ as an unknown in \eqref{eq:powerfit}; instead, we assume that it has been estimated and take it as prior information. The reason for this choice is that the SLODAR-based measurements do not seem to be sensitive to the value of $L_0$, and so it cannot be stably reconstructed. On the other hand, this also means that our method does not require an accurate estimate of $L_0$.

Another thing to note in \eqref{eq:powerfit} is that we still represent the PSD as a sum of radial basis functions, which does introduce some small amount of modeling error. This is done purely for computational reasons, since calculating the term corresponding to $\mb \vpsd(c, \gamma)$ exactly would require an expensive Fourier transform, which would take the runtime of \texttt{fmincon} from seconds to hours on a desktop computer.

\subsection{Method 2: Tikhonov Regularization with Power-Law Favoring Penalty}
\label{subsec:regularization}

Next, we consider two non-parametric regularization methods. These two-step methods do not assume a specific form of the power spectral 
density but instead stabilize the problem by adding an appropriate penalty term to the least-squares minimization problem. The difference between these last two methods is in the chosen penalty term.

In the first of these non-parametric methods, we aim to establish a (linear) Tikhonov-type regularization of the form 
\begin{equation}\label{eq:regularization}
\min_{\vrho, \vpsd \geq 0}  \left\{\norm{\ma \vrho + \mb \vpsd - b}_2^2 + \beta_1 \norm{\mgamma_1 (\vpsd - \vpsd_0)}_2^2 + \beta_2 \norm{\mgamma_2 (\vpsd - \vpsd_0)}_2^2 \right\}, \end{equation}
where $\mgamma_1$ and $\mgamma_2$ are regularization matrices representing our prior information about $\vpsd$ and the vector $\vpsd_0$ represents our \emph{a priori} estimate of $\vpsd$. The coefficients $\beta_1$ and $\beta_2$ are 
regularization parameters. 

Before we describe how $\mgamma_1$, $\mgamma_2$ and $\vpsd_0$ are chosen, let us first consider radial power spectral densities $\psdgr(\abs\vxi) = \psdg(\vxi)$ and $\psdr(\abs\vxi) = \psd(\vxi)$, where $\psd$ is the von K\'arm\'an PSD given in \eqref{eq:vonkarman}, and $\psdg$ is any radial and non-negative PSD which is continuously differentiable away from zero. Our goal is to design the regularization so that it favors PSDs $\psdg$ which are close to $\psd$ in the $H^1$-sense; in other words, we want both $\psdgr(r) - \psdr(r)$ and $\psdgr'(r) - \psdr'(r)$ to be small. However, as a power law, the magnitude of $\psdr(r)$ changes rapidly with $r$, and so it makes more sense to consider relative errors instead. Therefore, the goal of our regularization scheme is to favor functions $\psdgr(r)$ for which
\begin{equation}\label{eq:relerr}
\frac{\psdgr(r) - \psdr(r)}{\psdr(r)} \text{ and } \frac{\psdgr'(r) - \psdr'(r)}{\psdr'(r)} \text{ are small in $L^2$-norm.}
\end{equation}
It is important to emphasize that while this regularization scheme does strongly favor PSDs close to $\psdr$, it is still a non-parametric regularization that makes no assumptions on the specific form of $\psdgr$, and thus it allows much more freedom than the first two methods we have described in this section.

Note that even though we have assumed that $\psdr$ is the von K\'arm\'an PSD, there is no need to make that specific assumption if better prior information is available, as long as $\psdr$ and $\psdr'$ are nonzero everywhere. This is certainly true for power laws of the form given in \eqref{eq:gen_vonkarmanphase} with any negative exponent, so e.g.~the solution from method 1 may be used as the \emph{a priori} estimate of $\vpsd$; this is precisely what we have done in the numerical work we present in section \ref{sec:simulations}.

Now, as we are working in the discretized setting, we need to choose the regularization matrices $\mgamma_1$ and $\mgamma_2$ and the vector $\vpsd_0$ such that they give a discrete version of the relative error terms given in \eqref{eq:relerr}. To this end, let $\Delta_i = \abs{r_{i+1} - r_i}$ and define $\mk$ as the first-order difference operator given by
\begin{equation*}
K_{ij} = \begin{cases} - \frac{1}{\Delta_i}, & \text{when } j = i, \\
\frac{1}{\Delta_i}, & \text{when } j = i + 1, \\
0, & \text{otherwise,} \end{cases}
\end{equation*}
where $i = 1, \ldots, N_r-1$ and $j = 1, \ldots, N_r$. Thus $\mk$ is matrix with $N_r - 1$ rows and $N_r$ columns, and $\mk \vpsd$ approximates $\psdgr'(r)$ at the discretization points $r_i$.

Finally, to obtain the discretized form of \eqref{eq:relerr}, we define the regularization matrices as
\begin{equation}\label{eq:regmats}
\mgamma_1 = \diag(\vpsd)^{-1} \;\;\text{ and }\;\; \mgamma_2 = \diag(\mk \vpsd)^{-1} \mk,
\end{equation}
where $\diag(\vx)$ is the square diagonal matrix with the vector $\vx$ along its diagonal. Then the discrete analogy of \eqref{eq:relerr} is to favor discretized PSDs $\vpsd$ such that $\mgamma_1 (\vpsd - \vpsd_0)$ and $\mgamma_2 (\vpsd - \vpsd_0)$ are small in $L^2$-norm. This is precisely why the minimization in \eqref{eq:regularization} includes the squared $L^2$-norms of these two terms.

The vector $\vpsd_0$ can be obtained e.g.~by using the standard von K\'arm\'an PSD or by taking the solution given by method 1; in our numerical experiments we chose the latter option. The regularization parameters $\beta_1$ and $\beta_2$ determine the balance between trust in the measurements and in the prior information, and their values should be chosen appropriately. There is a vast amount of literature in the inverse problems field on how such parameters should be chosen. However, these parameter choice rules are beyond the scope of this paper; the values used in our numerical experiments were chosen by hand.

As a final remark, \eqref{eq:regularization} can be written as a quadratic programming problem, since the functional being minimized is quadratic in $\vpsd$. Indeed, if we denote $\mgamma := \beta_1 \mgamma_1^\T \mgamma_1 + \beta_2 \mgamma_2^\T \mgamma_2$ and set
\[ \mh := \begin{pmatrix} \ma^\T & \m0 \\ \m0 & \mb^\T \mb + \mgamma \end{pmatrix}, \;\; \vf := \begin{pmatrix} -\ma^\T \vb \\  - \mb^\T \vb - \mgamma \vpsd_0 \end{pmatrix}, \;\;  \vx := \begin{pmatrix} \vrho \\ \vpsd \end{pmatrix},  \]
then the minimization problem
\[ \min_{\vx \geq 0} \frac 12 \vx^T \mh \vx + \vf^T \vx \]
is equivalent to the minimization problem given in \eqref{eq:regularization}. We solved this problem using Matlab's \texttt{quadprog}, with the \texttt{'interior-point-convex'}-algorithm and default options except for an optimality tolerance of 100 times the floating point accuracy, i.e. approximately \num{2.22e-14}.

\subsection{Method 3: Total Variation Regularization with Power-Law Favoring Penalty}
\label{subsec:TVregularization}

Finally, we present a variant of the previous method where we impose a total variation (TV) prior on the derivative of the PSD and a Tikhonov-type regularization on the PSD itself. The purpose of the TV prior is to promote sparsity in the derivative term, which essentially means that the relative error between $\vpsd$ and the given prior PSD $\vpsd_0$ will be (close to) piecewise constant. From a physical point of view this may seem like an odd choice, but the goal is that this regularization might be better at pinpointing the parts of the PSD that deviate significantly from $\vpsd_0$. In other words, while the reconstruction may not have the visual appearance of a typical PSD, it should on the other hand be visually obvious where the reconstructed PSD differs from the prior assumption.

As before, the regularization is applied to the relative error of $\vpsd$ from a given prior PSD $\vpsd_0$ since the magnitude of power laws changes rapidly. The regularization problem then takes the form
\begin{equation}\label{eq:TVregularization}
\min_{\vrho, \vpsd \geq 0} \left\{\norm{\ma \vrho + \mb \vpsd - b}_2^2 + \beta_1 \norm{\mgamma_1 (\vpsd - \vpsd_0)}_2^2 + \beta_2 \norm{\mgamma_2 (\vpsd - \vpsd_0)}_1  \right\},
\end{equation}
where the only change from \eqref{eq:regularization} is that the second regularization term is changed from an $L^2$-norm to an $L^1$-norm, which is defined as
\begin{equation}\label{eq:L1norm}
\norm{\vx}_1 = \sum_{i=1}^n \abs{x_i}.
\end{equation}

The regularization matrices $\mgamma_1$ and $\mgamma_2$ are given by \eqref{eq:regmats}, as before. The prior estimate $\vpsd_0$ is also obtained the same way, by using the von K\'arm\'an PSD or the result obtained by using method 1. We again choose the latter option for our numerical experiments. The previous statements regarding the regularization parameters $\beta_1$ and $\beta_2$ also hold true for this method.

One important remark to make is that the problem given in \eqref{eq:TVregularization} does not at first glance seem to fit into the quadratic programming framework, since the $L^1$-norm defined in \eqref{eq:L1norm} involves an absolute value. However, the minimization problem can indeed be turned into a quadratic programming problem by clever use of auxiliary variables, as described in \cite{kolehmainen12}. We begin by defining $\mgamma_2(\vpsd-\vpsd_0) = \vu^+ - \vu^-$, where $\vu^+, \vu^- \geq 0$. Then \eqref{eq:TVregularization} may be equivalently written as
\[ \min_{\vrho, \vpsd, \vu^+, \vu^- \geq 0} \left\{\norm{\ma \vrho + \mb \vpsd - b}_2^2 + \beta_1 \norm{\mgamma_1 (\vpsd - \vpsd_0)}_2^2 + \beta_2 \v1^\T \vu^+ + \beta_2 \v1^\T \vu^- \right\}, \]
where $\v1$ is a vector of ones. Now, denote
\[ \vz = \begin{pmatrix} \vrho \\ \vpsd \\ \vu^+ \\ \vu^- \end{pmatrix}, \;\;\; \mq = \begin{pmatrix} \ma^\T \ma & \m0 & \m0 & \m0 \\ \m0 & \mb^\T \mb + \mgamma_1^\T \mgamma_1 & \m0 & \m0 \\ \m0 & \m0 & \m0 & \m0 \\ \m0 & \m0 & \m0 & \m0 \end{pmatrix}, \;\;\; \vc = \begin{pmatrix} -\ma^\T \vb \\ -\mb^\T \vb -\beta_1 \mgamma_1^\T \mgamma_1 \vpsd_0 \\ \beta_2 \v1 \\ \beta_2 \v1 \end{pmatrix}. \]
Then \eqref{eq:TVregularization} is equivalent to the quadratic programming problem
\begin{align*}
\text{minimize } &  \frac 12 \vz^\T \mq \vz + \vc^\T \vz \\
\text{subject to } & \mm \vz = \vm \text{ and } \vz \geq 0,
\end{align*}
where $\mm$ and $\vm$ encode the equality $\mgamma_2(\vpsd-\vpsd_0) = \vu^+ - \vu^-$:
\[ \mm := \begin{pmatrix} \m0 & \mgamma_1 & - \mi & \mi \end{pmatrix}, \;\;\; \vm = \mgamma_1 \vpsd_0. \]
As with method 2, we solved this problem using the \texttt{'interior-point-convex'}-algorithm in \texttt{quadprog} with an optimality tolerance of approximately \num{2.22e-14}.

\section{Numerical simulations}
\label{sec:simulations}

In this section we describe the numerical simulations used to test our methods, and present the results of these tests.

\subsection{Computing the system matrices}

The system matrices $\ma$ and $\mb$ of \eqref{eq:final} are expressed as Fourier transforms over $\R^2$. We estimate the integrals over $\R^2$ by integrals over the square domain $[-10,10]^2$, motivated by the decay rate shown in Theorem~\ref{thm:compact}. These finite Fourier integrals were evaluated using a method based on the Fast Fourier Transform (FFT). This method was first introduced by Bailey and Swarztrauber in \cite{Bailey94} and later generalized to the multidimensional case by Inverarity in \cite{Inverarity02}. Notice carefully that the number of integration points needs to be large enough to avoid any problems with the FFT related to aliasing \cite{Cooley67}.

\subsection{Simulating the atmosphere}
\label{sec:simdesc}
We studied the performance of our regularization methods by using MOST, a MATLAB tool for simulating adaptive optics systems which has been developed by the AAO team at JKU Linz in Austria. The atmosphere was simulated as discrete layers of turbulent air, using the standard von K\'arm\'an power spectral density given by \eqref{eq:vonkarman}, which for a layered atmosphere is given by
\[ \psdn(\vxi) = b \rho_k (\abs{\vxi}^2 + 1/L_0^2)^{-11/6}. \]
The outer scale above was set to $L_0 = \SI{25}{m}$. For the ground layer turbulence, we replaced the above PSD by the general version given by \eqref{eq:gen_vonkarman}, which in this case yields
\begin{equation}\label{eq:sim_psd}
\psdn(\vxi,\gamma) = b(\gamma) \rho_0 (\abs{\vxi}^2 + 1/L_0^2)^{-\gamma}.
\end{equation}
We chose the value $\gamma = 1.5732$ for the exponent. This was a matter of convenience: following the work of \cite{toselli2008} and \cite{li2015}, we found that $\gamma \approx 1.5732$ is the unique exponent which gives the same coefficient $b(\gamma)$ as we have in the standard von K\'arm\'an model with the exponent $\gamma_0 = 11/6 \approx 1.8333$.

In addition to the above, we also wanted to see how well we can recognize deviations from the von K\'arm\'an power law model. For this reason, we define the smooth bump function
\[ \Psi_{r_0}(r) = \begin{cases} \frac 12 \sin^2\left(\frac{(r-r_0+0.05)\pi}{0.1}\right), & \abs{r - r_0} < 0.05 \\ 0, & \text{otherwise.} \end{cases} \]
This function is supported in the interval $[r_0 - 0.05, r_0 + 0.05]$, attains the value $0.5$ at $r_0$ and is smooth everywhere. In our simulations we used a PSD with the same exponent $\gamma = 1.5732$ as above, but we added three bumps, centered at frequencies $0.4, 0.55, 0.7$; the full PSD was thus given by
\begin{equation}\label{eq:psd_bumps}
\psdn^\text{bumps}(\vxi) = b(\gamma) \rho_0 (\abs{\vxi}^2 + 1/L_0^2)^{-\gamma} (1 + \Psi_{0.4}(r) - \Psi_{0.55}(r) + \Psi_{0.7}(r)).
\end{equation}
This simulated PSD is shown in Fig.~\ref{fig:psd_bumps}, along with the reconstruction.

In the simulations we considered a telescope with a diameter of 42 meters and wavefront sensors with $84 \times 84$ subapertures; the subaperture size was thus $\ap = 0.5$ meters. The two LGSs were at directions $\vpsi = (3.75, 0)$ and $\vpsi' = (-3.75, 0)$, so the guide star separation was $\vtheta = (7.5, 0)$. These angles are given in arcminutes, and the coordinate system is fixed such that the rows and columns of subapertures in the WFSs are parallel to the $x$- and $y$-axes. The LGS separation $\vtheta$ is then parallel to the subaperture rows, so we choose the subaperture separations $\vd_k = (k,0)$, since these need to be in the same direction as $\theta$.

An important point to keep in mind is that for larger $k$, there are fewer pairs of subapertures separated by a distance $\ap \vd_k$, which in turn leads to a noisier estimate of the cross-correlation for $\vd_k$. We therefore restrict ourselves to the range $k = 0, \ldots, 60$ to ensure a sufficient quality of the measurements. Since we are using curvature measurements $m_x$ and $m_y$, which are weighted sums of adjacent subapertures, we have a grid of $82 \times 84$ measurements for $m_x$ and $84 \times 82$ measurements for $m_y$. This means that even for the separation $\vd_{60}$ we have $22 \times 84$ and $24 \times 82$ measurements of the cross-correlation in $m_x$ and $m_y$, respectively. Thus the cross-correlation at each time step for each $\vd_k$ is estimated as an average over at least 1600 measurements.

With $\vd_k$ chosen as above, the altitude for layer $k$ is obtained by drawing lines from two subapertures separated by a distance $\ap \vd_k$ to the corresponding LGSs, and choosing the altitude where these lines intersect, as in Fig.~\ref{fig:slodar}. These altitudes are given explicitly by the formula
\[ h_k = \frac{k \ap}{k \ap / H + \abs{\vtheta}}, \]
where the LGS altitude was $H = \SI{90}{km}$. This yields $h_0 = \SI{0}{m}$, $h_1 \approx \SI{229}{m}$ and $h_{60} \approx \SI{12}{km}$, with the remaining altitudes more or less evenly spaced.

The atmosphere was modelled as 61 discrete layers of turbulent air, located at the altitudes $h_k$ defined above. Thus the simulated and reconstructed layers are located at the same altitudes; this was done to ensure that the reconstructed and simulated turbulence profile can be easily compared. The simulated turbulence profile can be seen in Fig.~\ref{fig:cn2}. To avoid an inverse crime \cite{kaipio06}, we also consider a model with 9 layers which for the most part do not coincide with SLODAR layer altitudes.

We simulated 50 timesteps, regenerating the atmosphere each time to ensure that the measurements are statistically independent. This corresponds roughly to waiting until the atmosphere has moved over the telescope; more specifically, it ties into the rate of decorrelation of the atmospheric turbulence. In a study carried out by Guesalaga et al.~\cite{Guesalaga2014}, the authors found typical rates of decorrelation of \SIrange{1.0}{3.0}{\per\second}, although in the presence of heavy dome seeing rates as low as \SI{0.3}{\per\second} were obtained for the ground layer. The authors found a clear, seemingly linear correlation between wind speed and the rate of decorrelation. Thus wind speed and dome seeing appear to be the two key factors.

From the above, it seems reasonable to assume that a single simulated timestep takes roughly a second in the real world, or up to four seconds in the presence of dome seeing. This leads to an estimate of between one and four minutes for the time required to obtain 50 uncorrelated samples of the atmosphere, which is reasonable for turbulence profiling. Finally, since we have 50 timesteps and at least 1600 valid subaperture pairs for each $\vd_k$, this means we estimate the cross-correlations by averaging over more than 80 000 measurements.

During our simulations we also found that sufficient resolution for the turbulent layers is absolutely essential for accurate correlation measurements. As a trade-off between computation time and accuracy, we used matrices of size $8400 \times 8400$ for each layer; this produced an error of order \SI{1}{\percent} in the measurement vector $\vb$. We also found that there appears to be a limit on the frequency where data is available, since deviations such as the ones used in the PSD with bumps appear to be invisible to our methods if the bumps are located at frequencies greater than \SI{1}{\per\meter}. We suspect that this is due to the effect of aliasing, since the largest spatial frequency visible to Shack--Hartmann wavefront sensors should be of the order of $1/(2 D)$ due to the Nyquist sampling theorem. As the subaperture size is \SI{0.5}{\meter} in our simulations, this corresponds exactly to the frequency \SI{1}{\per\meter}.

\subsection{Computational costs}

The computational costs for the methods described in Section~\ref{sec:numerical_methods} fall into two separate categories. We have the offline cost of computing the system matrices $\ma$ and $\mb$, which can be quite heavy but can be done beforehand, and the online costs of determining the measurement vector $\vb$ and using the numerical solvers, which would in practice need to be done during telescope operation, although not on the same millisecond timescale as atmospheric tomography is required. All computation times cited here are for a laptop with an Intel Core i7-5600U processor, which has two cores running at \SI{2.6}{\giga\hertz} and two threads per core.

For the offline costs, computing a single column of the matrix $\ma$ or $\mb$ takes an average of 22 seconds. In our examples $\ma$ has 60 columns and $\mb$ has 401, so the total offline computation cost is just under three hours. The main component of this cost is that each column requires evaluating two Fourier transforms, which we solve using FFTs with the method given in \cite{Inverarity02}. 

The online cost consists of computing the correlations for the measurement vector $\vb$, and using the numerical solvers. The correlations for 50 samples of WFS measurements take 0.3 seconds to compute with our Matlab implementation. However, as we noted earlier, in the real world 50 uncorrelated samples of the atmosphere might correspond to between one and four minutes of measurements. At a sampling frequency of \SI{500}{\hertz} this yields between $30 000$ and $120 000$ samples, which would take 3 to 12 minutes on our laptop. 

Once the matrices $\ma$ and $\mb$ and the vector $\vb$ are available, finding the solution using method 1 or 2 takes \SI{1.8}{\second} on average, and method 3 requires \SI{3.5}{\second}; the times for methods 2 and 3 include the cost of running method 1 as a prior estimate, so we can see that the additional cost of solving the quadratic programming problem for method 2 is negligible. These times were obtained for the data where the simulated PSD had three bumps in it. These solutions are presented later in Figures \ref{fig:psd_bumps} and \ref{fig:psd_bumps_tv}.

It should be emphasised that both computing the matrices and the correlations can be trivially parallelized, and there is a lot of room for other optimizations as well since we did not put a lot of effort into speeding up the code. For example, the method given in \cite{Inverarity02} is not really necessary for this particular application, since it is far more flexible than we require, and this comes with an added computational cost. Specifically, the method uses two FFTs of $2N \times 2N$ matrices to reach the same accuracy which could in this application be reached with a single FFT of an $N \times N$ matrix, which would reduce the computation time by a factor of eight and the memory footprint by a factor of four.

\subsection{Numerical results for the 61-layer model}

We will use three error metrics for comparing the results: relative $L^2$-error for the turbulence profiles, average $L^2$-error of the logarithms for PSDs, and the residuals. The first two are defined by
\begin{equation}\label{eq:errors}
E_{C_n^2}(\vrho) = \frac{\norm{\vrho - \vrho_\text{sim}}_2}{\norm{\vrho_\text{sim}}_2}, \;\; E_{\text{PSD}}(\vpsi) = \frac{\norm{\log_{10}(\psdg(r)) - \log_{10}(\psd_\text{sim}(r))}_2}{r_{N_R-1}},
\end{equation}	
where $\psdg(r)$ is the function obtained from the discrete PSD $\vpsi$ as a sum of basis functions, $r_{N_R-1} = 10$ is the largest discretization point for the PSD, and the $L^2$-norm is computed through numerical integration. The reason for this is that the discretization points for the PSD have been chosen non-uniformly, with over a third of them located between 0 and 1.

The residual for the three methods we have presented in Section~\ref{sec:numerical_methods} is given by
\begin{equation}\label{eq:residual}
 E_{\text{res}}(\vrho, \vpsi) = \norm{\ma\vrho + \mb\vpsi - \vb}_2,
\end{equation}
and for the standard SLODAR method it is given by
\begin{equation}\label{eq:residual_slodar}
E_{\text{res}}(\vrho) = \norm{\ma\vrho - \vb}_2.
\end{equation}
Note that the turbulence profile errors $E_{C_n^2}(\vrho)$ will never include the ground layer, because there is no obvious way of determining the turbulence strength at the ground when the PSD differs from the von K\'arm\'an model used for all other layers. The ground layer is of course also omitted for results obtained with standard SLODAR to ensure that the presented error values are comparable. 

Let us first consider the simulated PSD without bumps. Fig.~\ref{fig:psd} shows a semilogarithmic plot of the simulated PSD and the reconstruction obtained by using method 1, i.e., the parametric method that fits a power law to the measurements. The von K\'arm\'an PSD given by \eqref{eq:vonkarman} is shown for reference, both to demonstrate the difference between it and the simulated PSD, and to emphasize how close the reconstruction is to the simulated PSD. The reconstructed exponent is \num{-1.5338}, which is a bit off from the simulated one of \num{-1.5732}, but this is quite reasonable given the limited number of samples.

\begin{figure}
				 \includegraphics[width=0.9\textwidth]{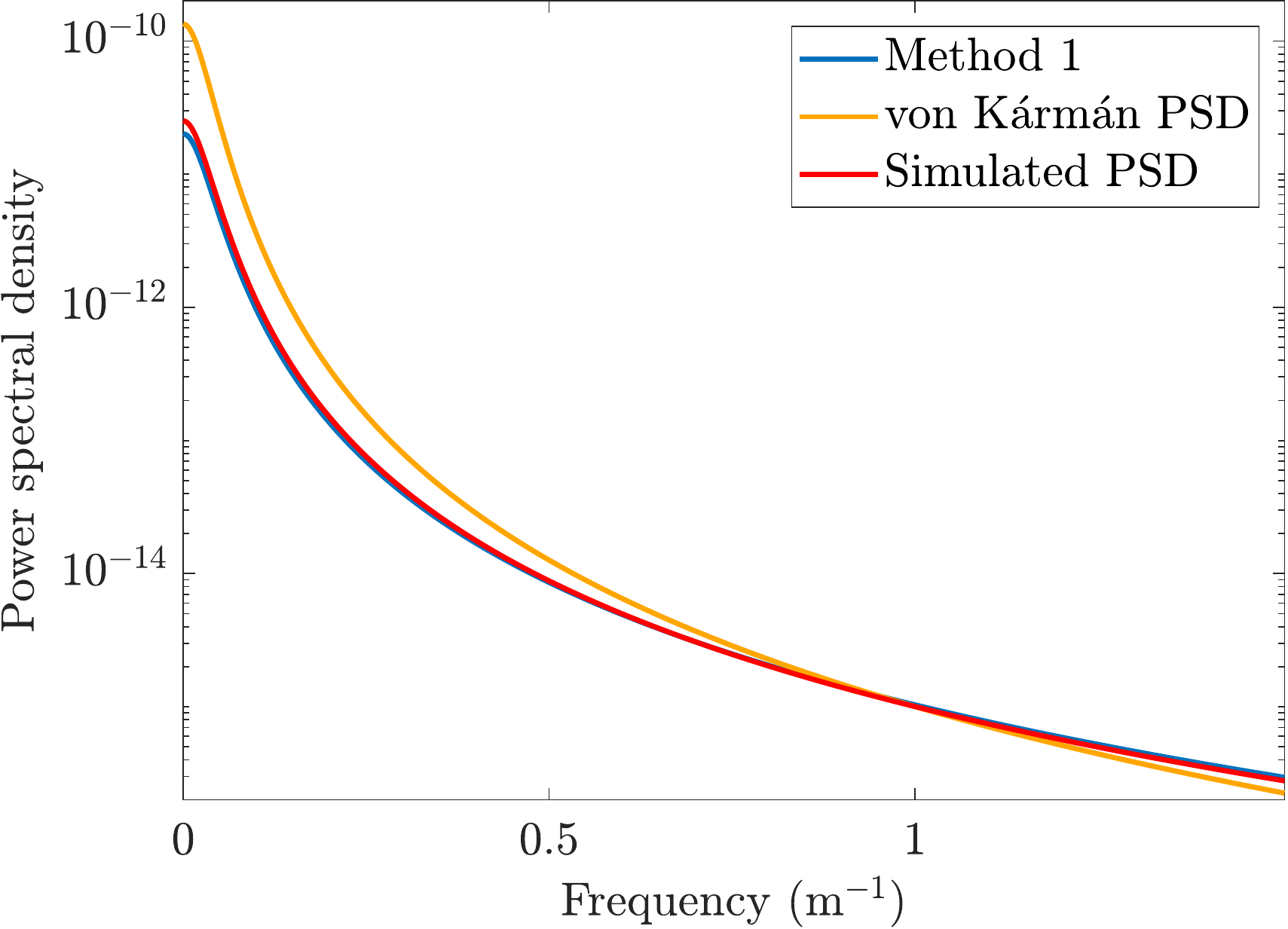}
				 \caption{Reconstruction (blue) of the ground layer PSD without bumps using method 1 (power law fitting), with the von K\'arm\'an PSD (yellow) and the simulated PSD (red) shown for reference. The reconstructed exponent was \num{-1.5338}, which is off by \SI{2.5}{\percent} from the true value of \num{-1.5732}.}
				 \label{fig:psd}
\end{figure}

In Fig.~\ref{fig:cn2}, we show the turbulence profile reconstructed by method 1 for layers 3--61. Our method agrees perfectly with standard SLODAR at the higher altitudes where the ground layer PSD has little or no effect, but at the layers close to the ground we see a much better agreement with the simulated turbulence profile than we would with standard SLODAR. The second layer is not shown in Fig.~\ref{fig:cn2} for graphical reasons, since the simulated turbulence strength is over four times as large as on the third layer. 

\begin{figure}
				 \includegraphics[width=0.9\textwidth]{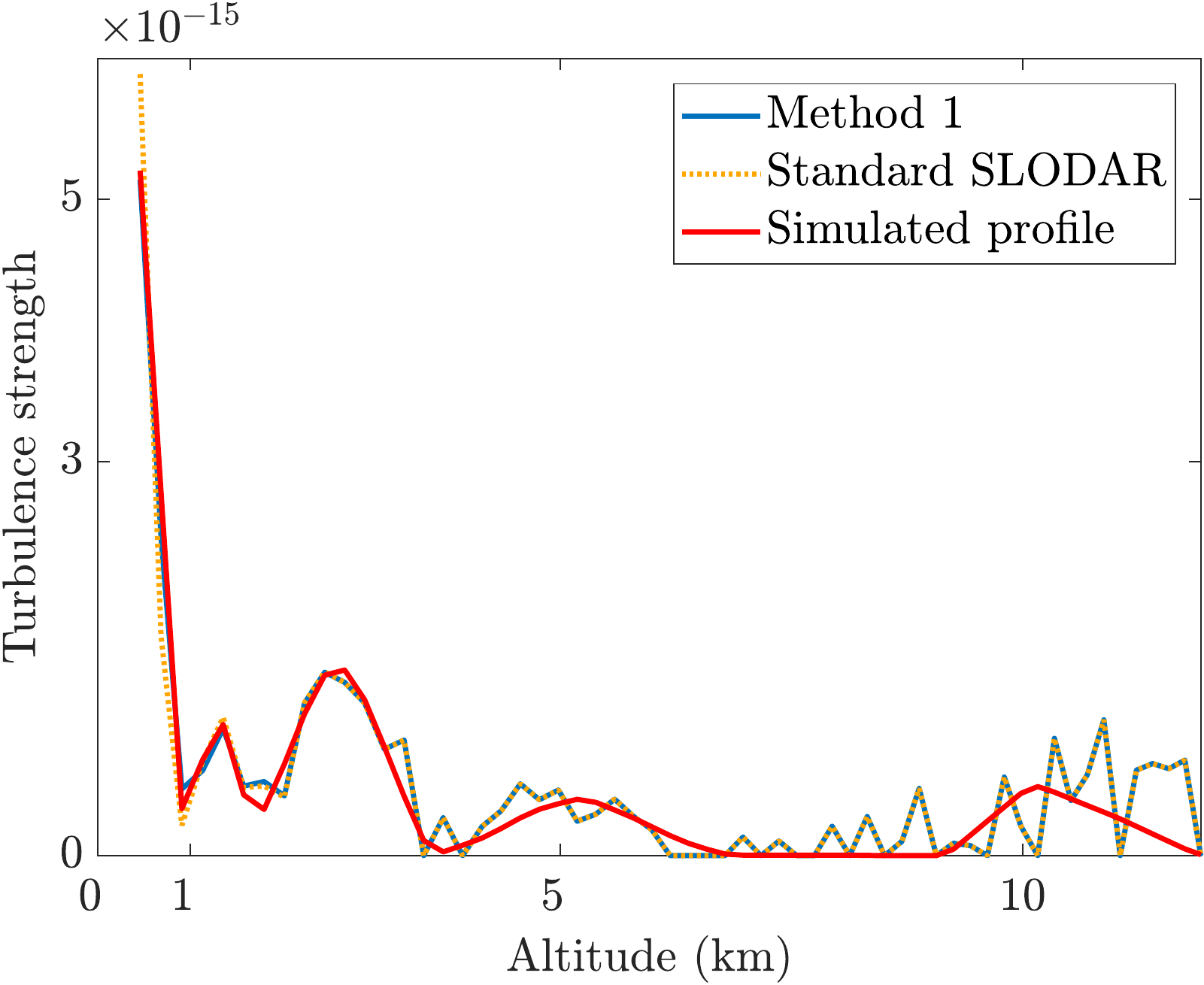}
				 \caption{Turbulence profile (blue) with method 1, with data generated using the simulated PSD without bumps. The true turbulence profile (red) is shown for reference, as well as the profile that would be obtained using standard SLODAR without PSD reconstruction (orange dots). The first two layers are not shown, since our method does not recover just the turbulence strength at the ground layer, and the magnitude of the second layer would dominate the figure. In the second layer, our method reduces relative error by a factor of four.}
\label{fig:cn2}
\end{figure}

The error metrics for all cases we will consider with the 61-layer atmosphere model have been collected into Table~\ref{table:errors}. As the second layer is omitted from turbulence profile images, we have also collected the turbulence profile values and relative errors for layer 2 into the same table. The top part of the table corresponds to Fig.~\ref{fig:cn2}, and as we can see, the relative error in the second layer drops roughly by a factor of four when using method 1, and the relative $L^2$-error for the turbulence profile is reduced by \SI{39}{\percent} compared to standard SLODAR.

\begin{table}[tp]
\caption{Values for the error metrics given in \eqref{eq:errors}, \eqref{eq:residual} and \eqref{eq:residual_slodar} for the 61-layer atmosphere in the three cases we will consider: the PSD without bumps given by \eqref{eq:sim_psd}, and the PSD with bumps given by \eqref{eq:psd_bumps} for both 50 and 500 samples of the atmosphere. Recall that $E_{\text{res}}$ is the $L^2$-residual, $E_{C_n^2}$ is the relative $L^2$-error for turbulence profiles and $E_\text{PSD}$ is the average logarithmic $L^2$-error for PSDs. We also include the turbulence profile values and relative errors in the second layer, as this is omitted in the figures in this section due to graphical reasons.}
\vspace{2mm}
\centering
\begin{tabular}{@{}lccccc@{}} \toprule
  & \multicolumn{3}{c}{Error metrics} &\multicolumn{2}{c}{Second layer of $C_n^2$-profile} \\ 
\cmidrule(r){2-4} \cmidrule(r){5-6} 
Method/PSD & $E_{\text{res}}$ ($10^{-14}$) & $E_{C_n^2}$ (\si{\percent}) & $E_{\text{PSD}}$ & Value ($10^{-14}$) & Rel.~error (\si{\percent})  \\ \midrule
No bumps, 50 samples  &    &    &    &    &    \\
\cmidrule(r){1-1}
Truth  &  5.4923  &  \NA  &  \NA  &  2.2757  &  \NA  \\
SLODAR  &  4.4218  &  13.76  &  \NA  &  2.0428  &  -10.23  \\
Method 1    &  2.6170  &  8.42  &  0.0660  &  2.3333 &  2.53  \\\midrule
Bumps, 50 samples  &    &    &    &    &    \\
\cmidrule(r){1-1}
Truth  &  4.3314  &  \NA  &  \NA  &  2.2757  &  \NA  \\
SLODAR  &  3.8235  &  12.59  &  \NA  &  2.0541  &  -9.74  \\
Method 1    &  2.3615  &  8.07  &  0.0268 &  2.2986 &  1.01  \\
Method 2  &  1.7796  &  6.98  &  0.0193  &  2.2913  &  0.68  \\
Method 3  &  1.8015  &  7.16  &  0.0207  &  2.2933  &  0.77  \\\midrule
Bumps, 500 samples  &    &    &    &    &    \\
\cmidrule(r){1-1}
Truth  &  1.8485  &  \NA  &  \NA  &  2.2757  &  \NA  \\
SLODAR  &  3.7227  &  11.01  &  \NA  &  2.0547  &  -9.71  \\
Method 1    &  1.5628  &  5.61  &  0.0422 &  2.3282 &  2.30  \\
Method 2  &  0.6569  &  2.88  &  0.0358  &  2.3183  &  1.87  \\
Method 3  &  0.6829  &  3.05  &  0.0353  &  2.3182  &  1.86  \\
\bottomrule
\end{tabular}
\label{table:errors}
\end{table}

Next, we consider the PSD with bumps given by \eqref{eq:psd_bumps}. In Fig.~\ref{fig:psd_bumps}, we can see the simulated PSD and the reconstructions obtained using methods 1 and 2; note that the solution from method 1 serves as the prior estimate for method 2. The von K\'arm\'an PSD is not shown in Fig.~\ref{fig:psd_bumps}, but recall that the simulated PSD here uses the same exponent as the simulated PSD in Fig.~\ref{fig:psd}; the only difference between the two is the presence of three bumps in the PSD. 

We can see from Fig.~\ref{fig:psd_bumps} that method 1 picks out the overall trend of the PSD very well; in fact, the reconstructed exponent in this case is \num{-1.5886}, which is better than what we had in the previous example. The reconstruction shown in Fig.~\ref{fig:psd_bumps} was obtained with method 2, using values of $\beta_1 = \num{5e-6}$ and  $\beta_2 = \num{6e-7}$ for the regularization parameters. The method can find all three bumps, although they are somewhat wider than they should be.

\begin{figure}
				 \includegraphics[width=0.9\textwidth]{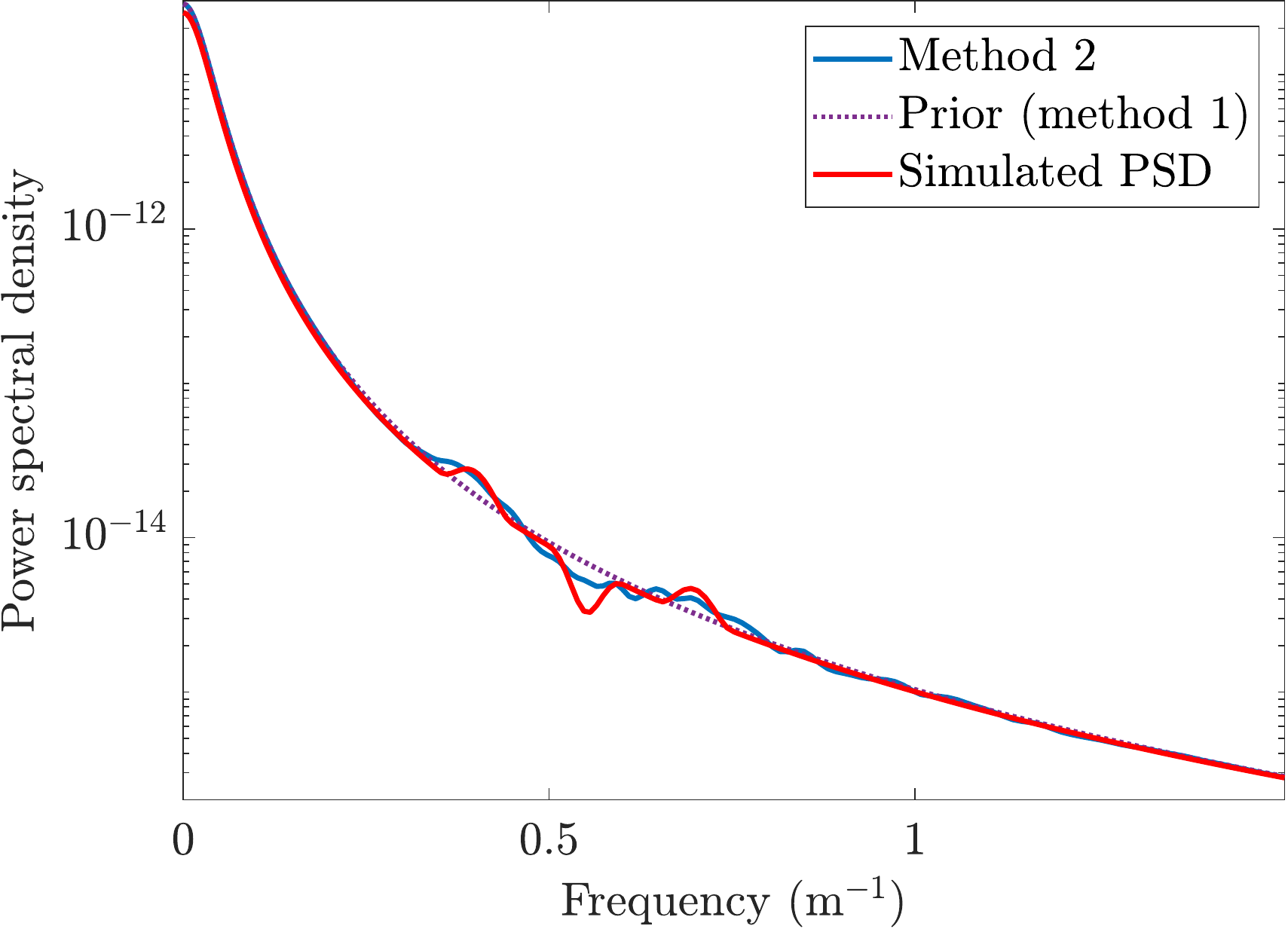}
				 \caption{Reconstruction (blue) of the ground layer PSD with bumps using method 2 (Tikhonov regularization) with $\beta_1 = \num{5e-6}$ and $\beta_2 = \num{6e-7}$. The simulated PSD (red, given by \eqref{eq:psd_bumps}) and the \emph{a priori} estimate using method 1 (purple dots) are shown for reference. The reconstructed exponent using method 1 was \num{-1.5886}, which gives a relative error of just \SI{1}{\percent}. Note that while the von K\'arm\'an PSD is not shown here, the difference between it and the simulated PSD is the same as it was in Fig.~\ref{fig:psd}.}
				 \label{fig:psd_bumps}
\end{figure}

The average logarithmic $L^2$-errors for PSDs can be seen in the middle part of Table~\ref{table:errors}. For comparison, the error we would obtain by comparing the PSD without bumps given by \eqref{eq:sim_psd} to the PSD with bumps given by \eqref{eq:psd_bumps} is \num{0.0234}; the value given by method 2 is almost \SI{18}{\percent} smaller than this.

Fig.~\ref{fig:cn2_bumps} shows the reconstructed turbulence profile using method 2, corresponding to the PSD reconstruction in Fig.~\ref{fig:psd_bumps}. The reconstruction given by our method is not quite as good as it was in Fig.~\ref{fig:cn2}, but we can also see the standard SLODAR method having much more trouble, with two layers close to the ground being reconstructed as zero. As before, the first layer shown in Fig.~\ref{fig:cn2_bumps} is the third layer in the atmosphere since the ground layer turbulence strength cannot be shown and the second layer was removed due to graphical reasons.

\begin{figure}
				 \includegraphics[width=0.9\textwidth]{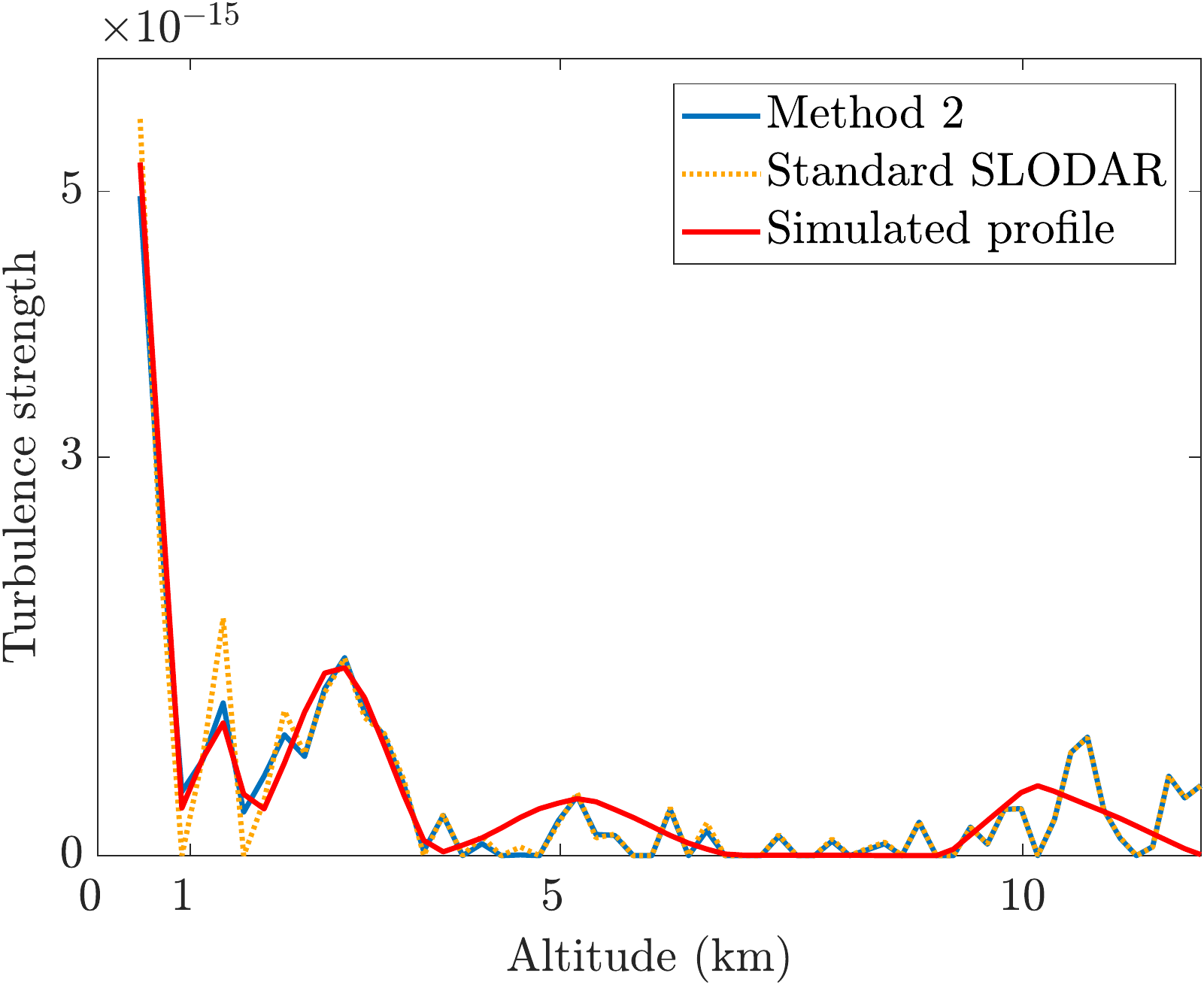}
				 \caption{Reconstruction (blue) of the turbulence profile using method 2 (Tikhonov regularization) on the simulated PSD with bumps; this corresponds to the reconstructed PSD shown in Fig.~\ref{fig:psd_bumps}. The true turbulence profile (red) and the reconstruction using the standard SLODAR method (orange dots) are shown for reference. As in the turbulence profile shown in Fig.~\ref{fig:cn2}, the first layer cannot be shown since the turbulence strength of that layer is coupled with the PSD shown in Fig.~\ref{fig:psd_bumps}, and the second layer is omitted for graphical reasons. This time, however, our reconstruction is much closer to the true profile in the second layer than the standard SLODAR method.}
				\label{fig:cn2_bumps}
\end{figure}

Looking at the middle part of Table~\ref{table:errors}, we see an even bigger improvement in the second layer than we did before. The relative error of SLODAR is still close to \SI{10}{\percent}, while all three of our methods give errors of around \SI{1}{\percent}. This highlights the benefits of recovering the ground layer PSD and not just the turbulence profile. The relative $L^2$-errors for the turbulence profiles show similar results, although here the difference is less pronounced since the higher altitudes are still quite noisy and are unaffected by the ground layer reconstruction. We can also see the residuals dropping steadily as the method complexity increases, and all residuals are smaller than the residual for the ''truth'' in Table~\ref{table:errors}, which was obtained by evaluating \eqref{eq:residual} with the true PSD and turbulence profile. It is also positive to see that the residuals we obtain are not significantly smaller than the true value, as that would be an indication that we are fitting to noise.

Fig.~\ref{fig:psd_bumps_tv} shows the reconstruction obtained for the simulated PSD with bumps using method 3, which uses total variation regularization. This time we have focused the image on the region with bumps, since the reconstruction outside this region is virtually identical to the one shown in Fig.~\ref{fig:psd_bumps} which was obtained with method 2. We can see that the reconstruction looks like a power law with sudden jumps in magnitude -- this is precisely the goal of method 3, since the total variation favors a piecewise constant relative error between the \emph{a priori} estimate and the reconstruction. This enables us to see the location of the bumps very clearly. We omit the turbulence profile for this reconstruction since the result is very close to the one shown in Fig.~\ref{fig:cn2_bumps}.

\begin{figure}
				 \includegraphics[width=0.9\textwidth]{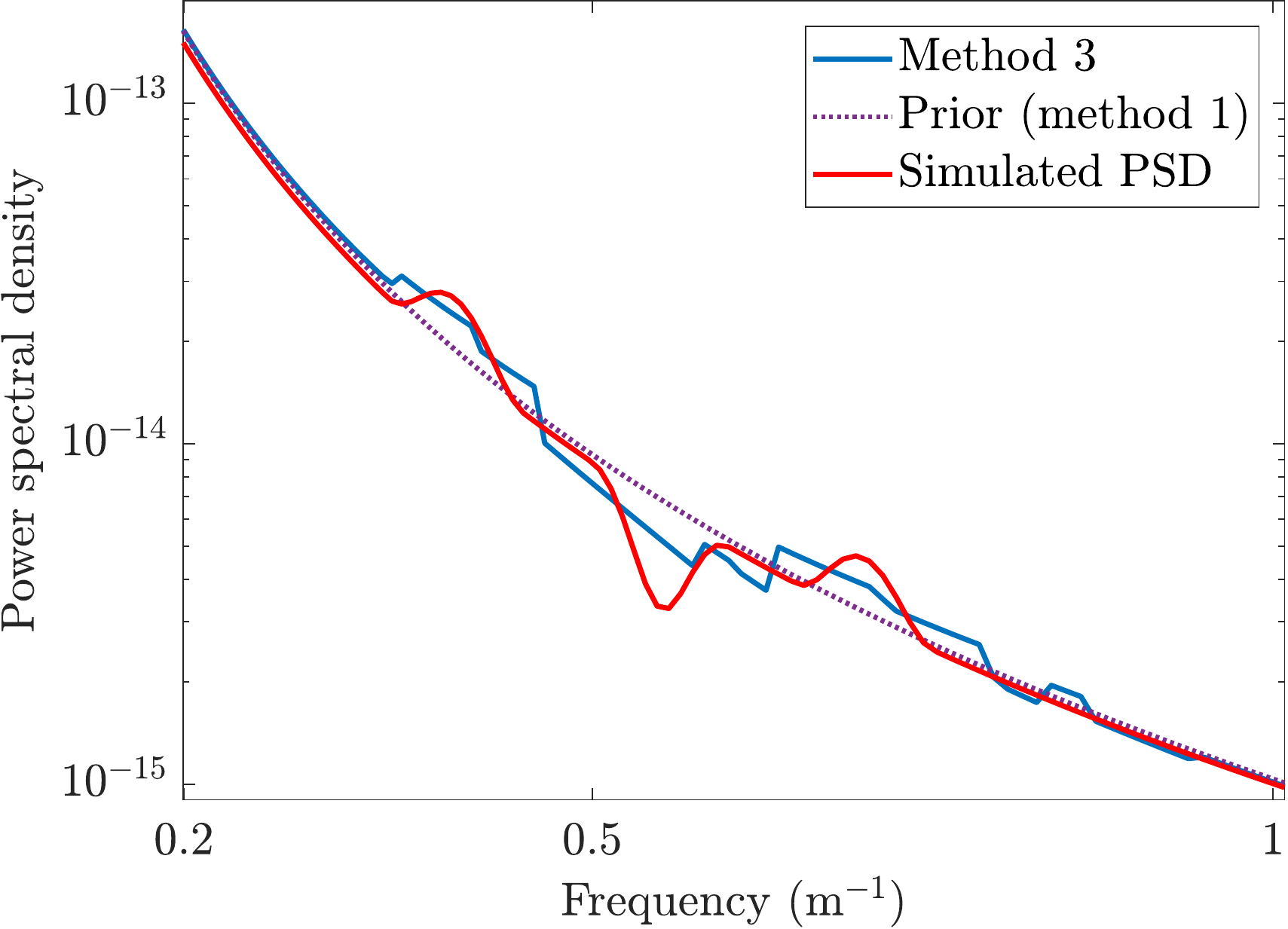}
				 \caption{Reconstruction (blue) of the ground layer PSD with bumps using method 3 (total variation regularization) with $\beta_1 = \num{5e-6}$ and $\beta_2 = \num{9e-7}$, with the simulated PSD (red) and the \emph{a priori} estimate using method 1 (purple dots) shown for reference. The image has been focused on the region with bumps, since there is virtually no difference in the reconstructions obtained by the three methods outside this region.}
				 \label{fig:psd_bumps_tv}
\end{figure}

Looking again at Table~\ref{table:errors}, we can compare the results from methods 2 and 3 to find that method 3 does slightly worse in every error metric, although the difference is quite insignificant. Overall, the biggest difference between methods 2 and 3 seems to be the visual appearance rather than numerical accuracy; it could even be argued that method 3 in this case locates the bumps more clearly, even though its accuracy is worse.

It is of course also interesting to see if we can improve reconstructions with a greater number of samples. In Fig.~\ref{fig:psd_bumps_500} we show the same results using method 2 as we did in Fig.~\ref{fig:psd_bumps}, but we have now included a second reconstruction which uses 500 timesteps of data, rather than the 50 timesteps we have used for all other results. We can see that most of the problems with the reconstruction from 50 timesteps of data are now gone and the bumps are reconstructed quite nicely. There are still some artefacts such as the small extra bump after the third bump, but overall the reconstruction has improved quite a bit. The exponent given by method 1 is now \num{-1.5572}, which is again \SI{1}{\percent} off from the true value. 

\begin{figure}
				 \includegraphics[width=0.9\textwidth]{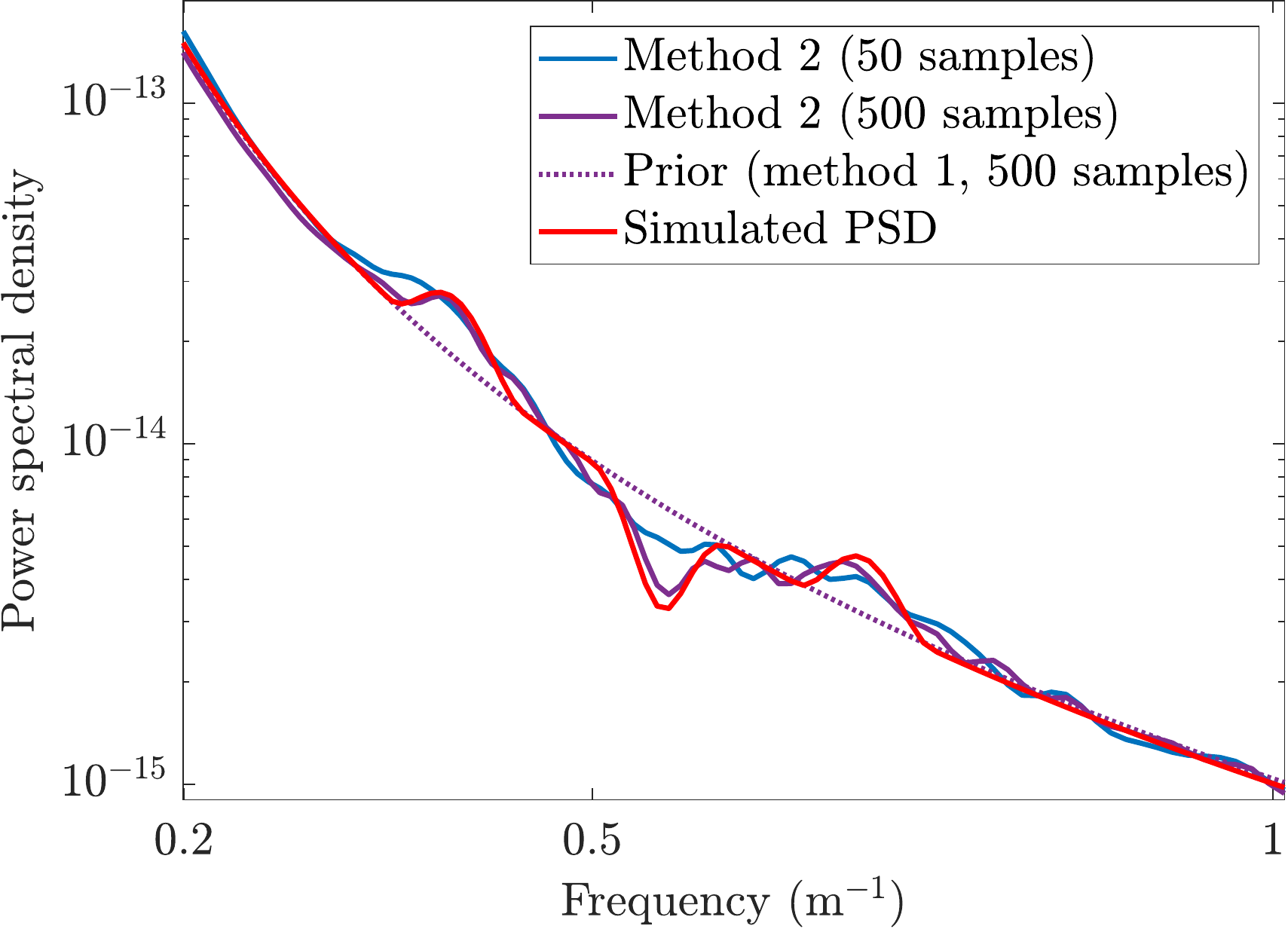}
				 \caption{Reconstruction of the ground layer PSD with bumps using method 2 (Tikhonov regularization). The figure is focused on the region with bumps, and demonstrates the difference between using 50 (blue) or 500 (purple) samples. The simulated PSD (red) and the \emph{a priori} estimate (purple dots) are again shown for reference. The bumps are reconstructed much more clearly with 500 samples, although there are still some artefacts visible. The reconstructed exponent using method 1 was \num{-1.5572}, which again gives a relative error of just \SI{1}{\percent}. The values for $\beta_1$ and $\beta_2$ are \num{5e-6} and \num{6e-7} for 50 samples (as before), and \num{1e-6} and \num{1e-7} for 500 samples.}
\label{fig:psd_bumps_500}
\end{figure}

Fig.~\ref{fig:psd_bumps_rel_err} shows the reconstructions from Fig.~\ref{fig:psd_bumps_500} in a different way, by showing the relative errors of each reconstruction compared to the simulated PSD. From Fig.~\ref{fig:psd_bumps_rel_err} we can really see that the region around the bumps is reconstructed much better with 500 timesteps worth of data; the relative error is much flatter than with 50 timesteps. Before the first bump we can see the presence of an extra bump that wasn't there with 50 timesteps, but that appears to be almost the only part where the error becomes worse.

\begin{figure}
				 \includegraphics[width=0.9\textwidth]{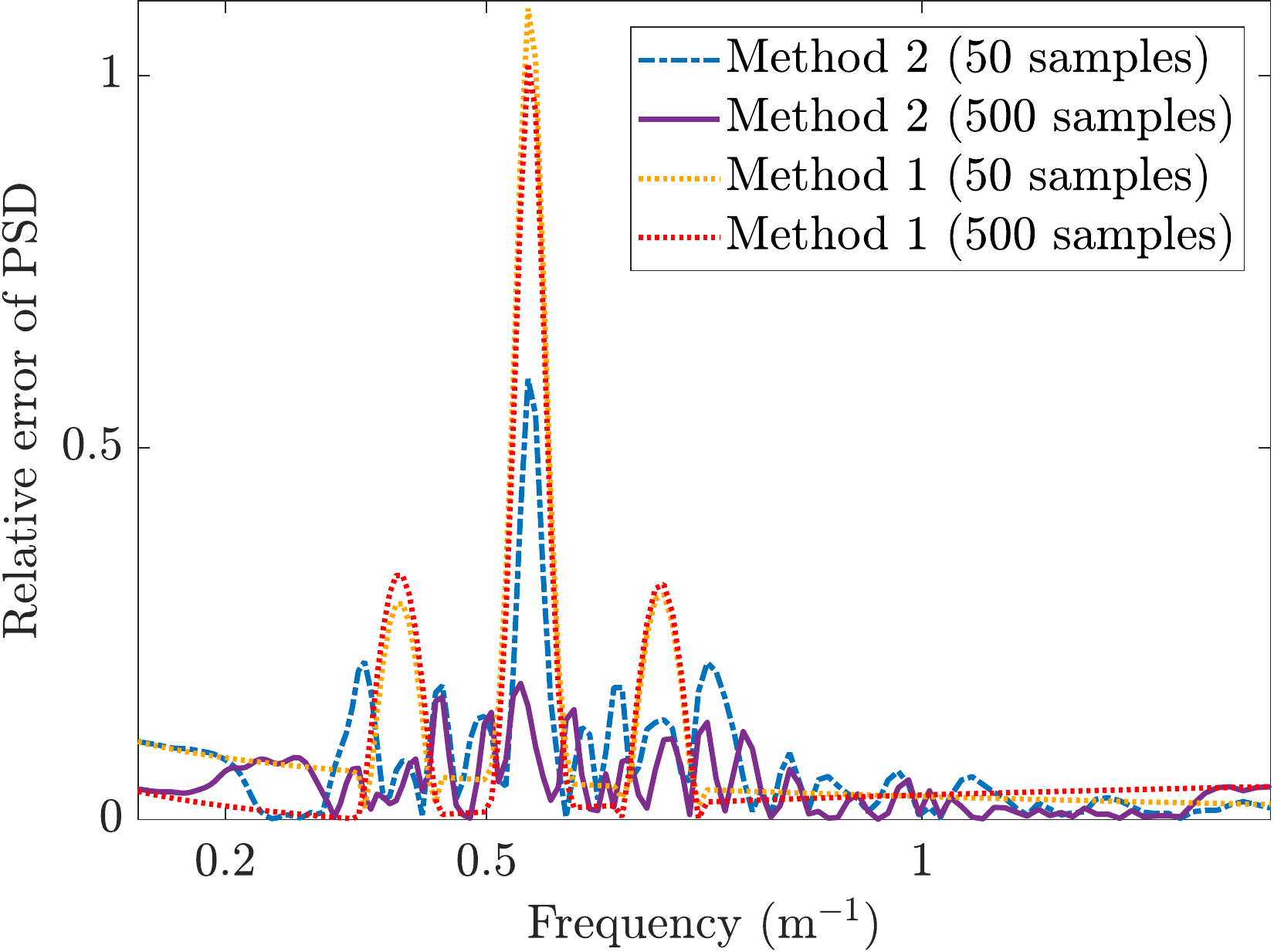}
				 \caption{Relative errors between the simulated PSD and the reconstructions obtained for 50 and 500 samples with methods 1 and 2. The figure is focused on the region where the reconstructions with method 2 differ significantly from the ones with method 1.}
				\label{fig:psd_bumps_rel_err}
\end{figure}

Finally, in Fig.~\ref{fig:cn2_bumps_500} we show the turbulence profile reconstructed by using method 2 with 500 samples, corresponding to the PSD shown in purple in Fig.~\ref{fig:psd_bumps_500}. This time the potential benefit of our method is much clearer than it was in Fig.~\ref{fig:cn2_bumps} -- the reconstruction of the first ten layers is now almost perfect, while the standard SLODAR method still struggles with them. As expected, the reconstructions at higher altitudes are also less noisy for both methods due to the increased number of samples.

\begin{figure}
				 \includegraphics[width=0.9\textwidth]{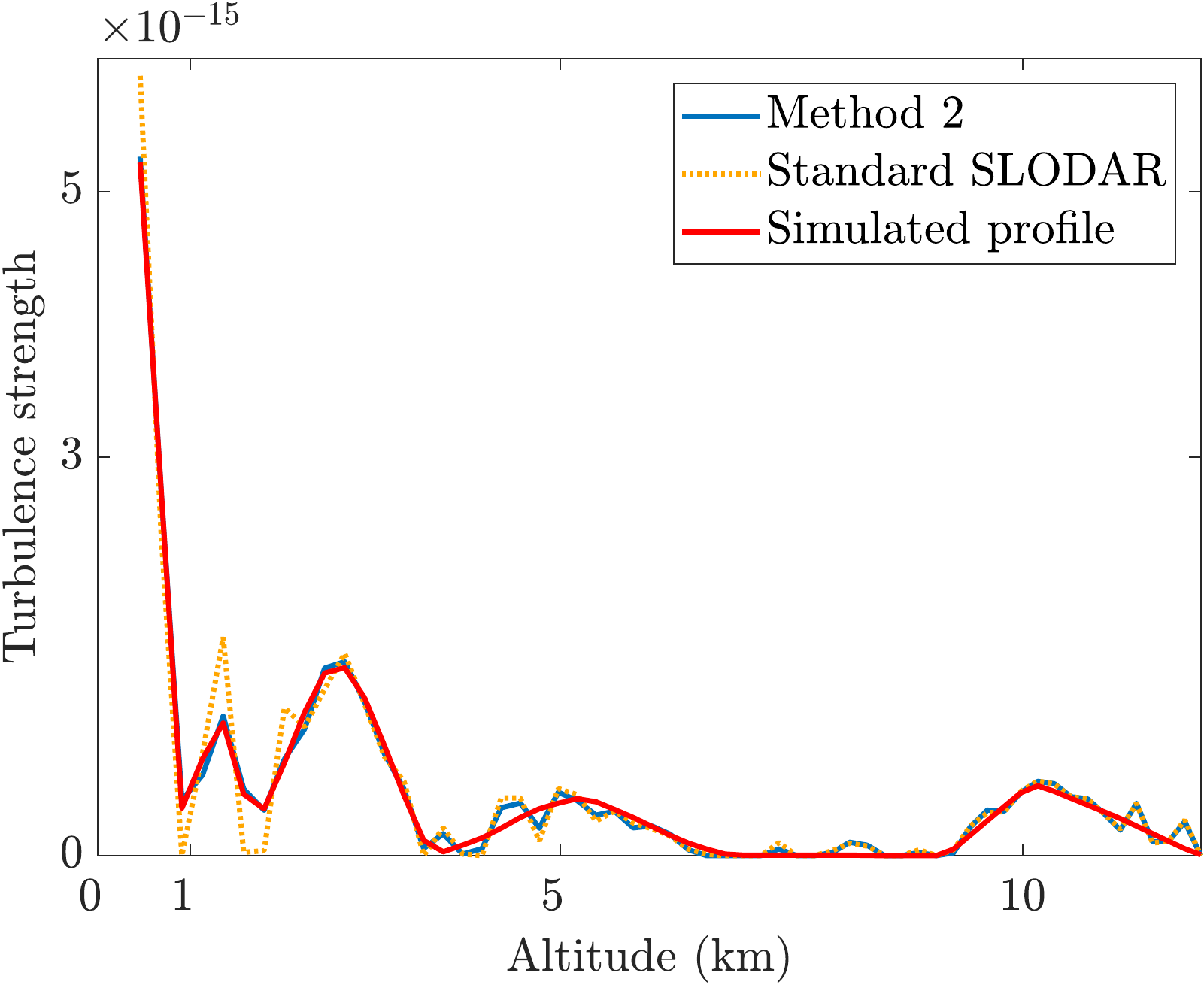}
				 \caption{Reconstruction (blue) of the turbulence profile using method 2 (Tikhonov regularization) with 500 samples on the simulated PSD with bumps; this corresponds to the reconstructed PSD shown in purple in Fig.~\ref{fig:psd_bumps}. The true turbulence profile (red) and the reconstruction using the standard SLODAR method (orange dots) are shown for reference. As in Fig.~\ref{fig:cn2_bumps}, the first two layers are omitted, but our reconstruction for the second layer is even closer to the truth than it was with 50 samples, whereas the standard SLODAR method actually gets worse results than with 50 samples.}
				\label{fig:cn2_bumps_500}
\end{figure}

Going back to Table~\ref{table:errors}, the bottom part shows the results for 500 timesteps; the regularisation parameters for method 3 with 500 samples were $\beta_1 = \num{1e-6}$ and $\beta_2 = \num{2e-7}$. Comparing to the results with 50 timesteps, we find that the ''true'' residual has dropped by a factor of almost \num{2.5}. This is more or less what we would expect, as in principle the noise level in the measurements should depend on the inverse square root of the number of measurements; thus, as we increase the number of measurements by a factor of 10, the noise should be reduced by roughly a factor of 3. The surprising thing about the numbers in the table is that the residual for SLODAR hardly drops at all, and is in fact larger than the residual we obtain using the simulated PSD and turbulence profile. This could be an indication that its residual is dominated by the modeling error caused by using the wrong PSD. The residual for method 1 is also reduced by only \SI{34}{\percent}, which is nowhere near the reduction of the true residual; this is not surprising given that the relative error in the PSD exponent is the same as with 50 samples. It is also possible that method 1, like SLODAR, starts reaching its limits in terms of the residual, as it cannot reconstruct the bumps in the PSD. For methods 2 and 3, however, we see a similar reduction as for the true residual. This is likely a result of the much better reconstruction for the bumps; we have not included the plots for method 3 with 500 samples, but the change from 50 samples is very similar to what we saw for method 2 in Fig.~\ref{fig:psd_bumps_500}.

Surprisingly, Table~\ref{table:errors} shows that both the PSD errors and the values for the second turbulent layer are worse than with 50 samples. This is however most likely due to the inherent variance in using only 50 samples, as there is no reason to think that less samples would in general yield better results. The change in PSD error when moving from method 1 to methods 2 and 3 is also slightly smaller, although this is mainly due to the slight deviation from the simulated PSD at the low frequencies in Fig.~\ref{fig:psd_bumps}, as the error in reconstructing the bumps is significantly smaller with 500 samples. For the turbulence profiles we see improvement across the board in the relative errors. These in fact seem to follow the same trend as the residuals, with slight improvements for SLODAR and method 1 and significant improvements for methods 2 and 3.

\subsection{Numerical results for the 9-layer model}

In this section we will show similar results with a 9-layer model to demonstrate that the quality of the results is not due to the choice of layer altitudes. Here, we only consider the results for the simulated PSD with bumps, as this is the hardest case to solve. 

Fig.~\ref{fig:cn2_bumps_9layer} is analogous to Fig.~\ref{fig:cn2_bumps} from the previous section, showing the reconstructed turbulence profile using method 2, for 50 samples of the 9-layer atmosphere. The true profile is this time shown using bars rather than a curve, because the layer altitudes are not as evenly spaced as in the 61-layer case. We can see that both standard SLODAR and our method pick out the layers at higher altitudes very well, with turbulence strength split between two layers whenever the true layer altitude lies between two SLODAR layers. Note also that the second layer of the simulated turbulence profile is located at \SI{140}{m} while the second reconstructed layer is at roughly \SI{228}{m}, which may impact the quality of ground layer reconstructions.

\begin{figure}
				 \includegraphics[width=0.9\textwidth]{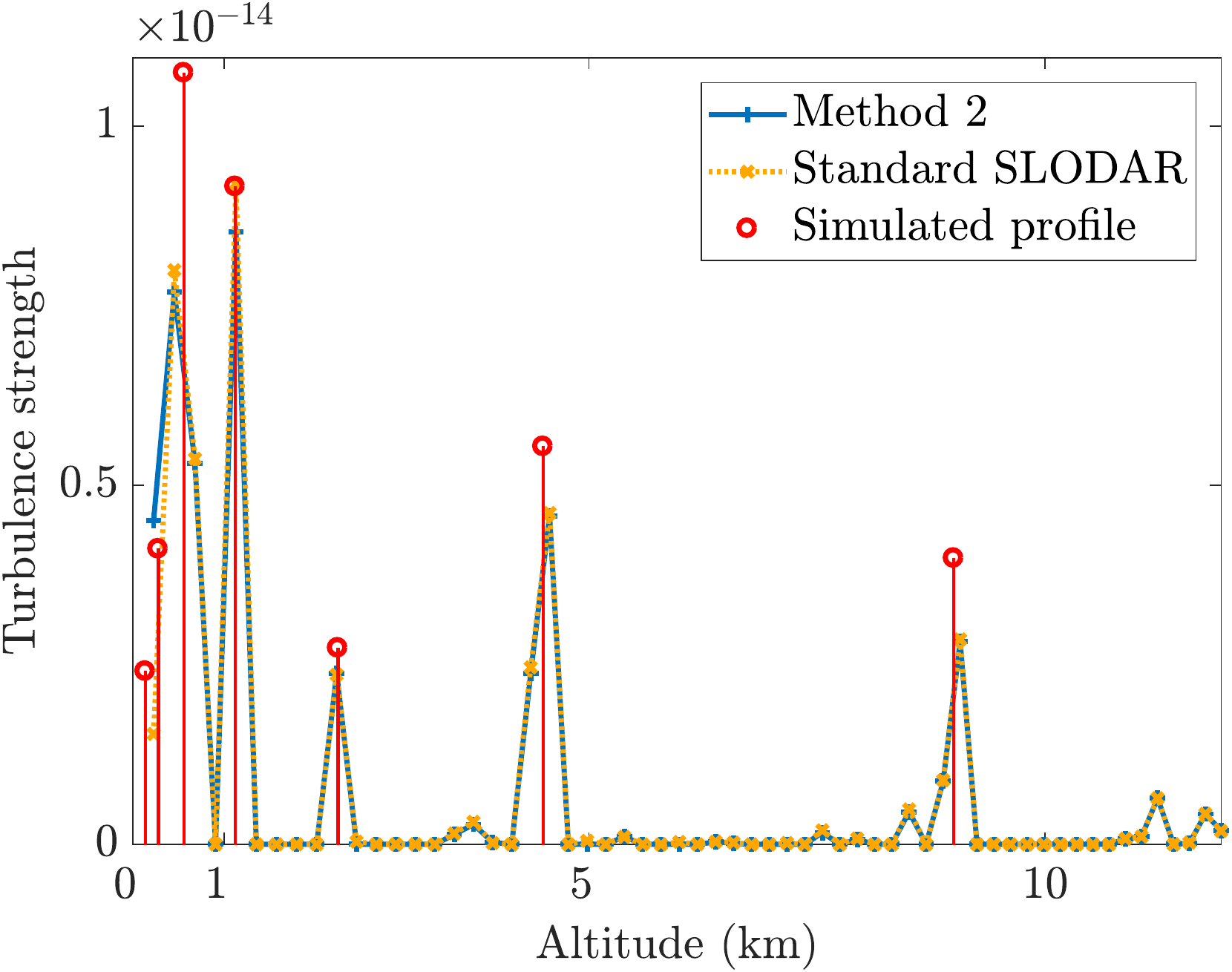}
				 \caption{Reconstruction (blue) of the turbulence profile using method 2 on the simulated PSD with bumps, using 50 samples with the 9-layer atmosphere. The true turbulence profile (red) and the reconstruction using the standard SLODAR method (orange) are shown for reference. The ground layer is omitted for each profile as before (but the second layer is included). The true profile also has a layer at \SI{18}{\kilo\meter} which is not shown here. Note that the reconstructed turbulence above \SI{10}{\kilo\meter} is likely only due to noise, as SLODAR-based methods are insensitive to turbulence located far above their maximum altitude, which in this case is roughly \SI{12}{\kilo\meter}.}
				\label{fig:cn2_bumps_9layer}
\end{figure}

The most notable feature of Fig.~\ref{fig:cn2_bumps_9layer}, however, is the stark difference between the values given by standard SLODAR and our method for the first reconstructed layer shown in the plot. Standard SLODAR predicts a very small value for this layer, even though there are two layers of the true profile with much larger turbulence strength very close to that layer. It is clear that our method gives a more accurate result close to the ground in this example.

Fig.~\ref{fig:psd_bumps_9layer} is analogous to Fig.~\ref{fig:psd_bumps_500}, showing the PSD reconstructions with method 2 for both 50 and 500 atmosphere samples at the same time. Somewhat surprisingly, the middle bump is actually reconstructed very well even with 50 samples, in stark contrast to the reconstruction shown in Fig.~\ref{fig:psd_bumps_500}. The reconstruction with 500 samples also follows the bumps even better than with the 61-layer model. However, the false bumps after the third bump are worse than previously, so much so that with 50 samples we almost miss the real third bump completely. Overall though, it is hard to say which case gives us better reconstructions.

\begin{figure}
				 \includegraphics[width=0.9\textwidth]{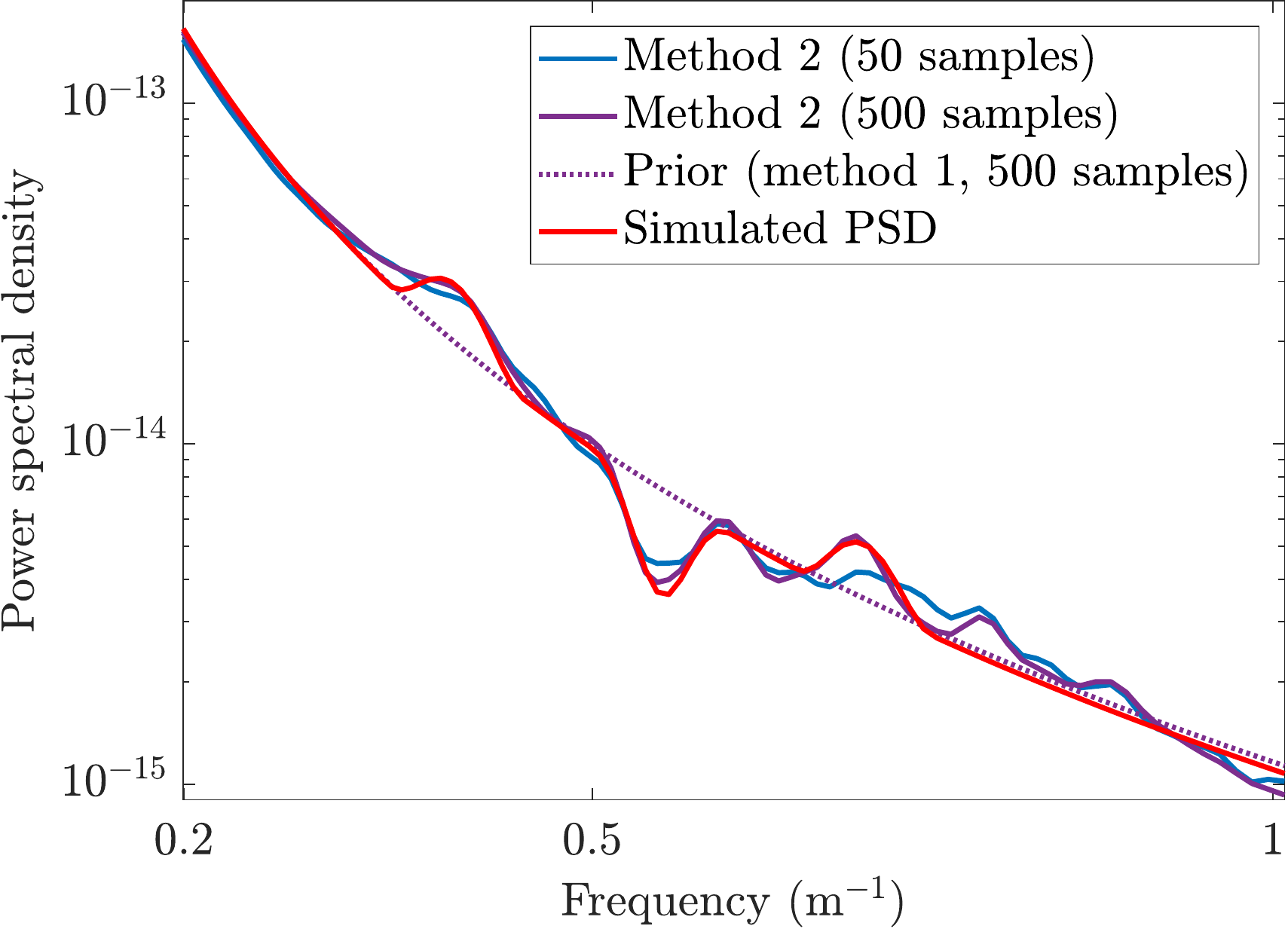}
				 \caption{Reconstruction of the ground layer PSD with bumps using method 2. The figure is focused on the region with bumps, and demonstrates the difference between using 50 (blue) or 500 (purple) samples, with the simulated PSD (red) and the \emph{a priori} estimate (purple dots) again shown for reference. The exponents given by method 1 are \num{-1.5428} for 50 samples and \num{-1.5518} for 500 samples; the corresponding relative errors are \SI{1.93}{\percent} and \SI{1.36}{\percent}, which is slightly worse than for the 61-layer model. The regularization parameters $\beta_1$ and $\beta_2$ are \num{5e-6} and \num{8e-7} for 50 samples, \num{3e-6} and \num{4e-7} for 500 samples.}
\label{fig:psd_bumps_9layer}
\end{figure}

Table~\ref{table:errors_9layer} shows the error metrics for the 9-layer atmosphere. This time we only include the residual and the PSD error, as the error in turbulence profiles cannot be computed when reconstructed and simulated layers are located at different altitudes, and there is no need to include the second layer anymore since it is visible in Fig.~\ref{fig:cn2_bumps_9layer}. The table also includes the results for method 3 and for the simulated PSD without bumps; we did not include plots for these, as they are very similar to what we saw with the 61-layer atmosphere. The regularization parameters $\beta_1$ and $\beta_2$ chosen for method 3 were \num{2e-5} and \num{3e-6} for 50 samples, \num{1e-5} and \num{2e-6} for 500 samples.

\begin{table}[tp]
\caption{Values for the error metrics given in \eqref{eq:errors}, \eqref{eq:residual} and \eqref{eq:residual_slodar} for the 9-layer atmosphere in the three cases we will consider: the PSD without bumps given by \eqref{eq:sim_psd}, and the PSD with bumps given by \eqref{eq:psd_bumps} for both 50 and 500 samples of the atmosphere. Recall that $E_{\text{res}}$ is the $L^2$-residual and $E_\text{PSD}$ is the average logarithmic $L^2$-error for PSDs. Note that the relative $L^2$-error for turbulence profiles is omitted here out of necessity, since the true profile uses different altitudes than SLODAR. The true residuals were obtained by computing a SLODAR matrix $\ma$ at the altitudes used by the 9-layer profiles.}
\vspace{2mm}
\centering
\begin{tabular}{@{}lcc@{}} \toprule
  & \multicolumn{2}{c}{Error metrics}\\ 
\cmidrule(r){2-3} 
Method/PSD & $E_{\text{res}}$ ($10^{-14}$) & $E_{\text{PSD}}$ \\ \midrule
No bumps, 50 samples  &    &    \\
\cmidrule(r){1-1}
Truth  &  4.6819  &  \NA  \\
SLODAR  &  5.3459  &  \NA  \\
Method 1    &  3.9969  &  0.0170  \\\midrule
Bumps, 50 samples  &    &    \\
\cmidrule(r){1-1}
Truth  &  5.4319  &  \NA  \\
SLODAR  &  7.2452  &  \NA  \\
Method 1    &  5.9790  &  0.0671  \\
Method 2  &  4.6707  &  0.0646  \\
Method 3  &  4.8081  &  0.0646  \\\midrule
Bumps, 500 samples  &    &    \\
\cmidrule(r){1-1}
Truth  &  1.9787  &  \NA  \\
SLODAR  &  5.6759  &  \NA  \\
Method 1    &  4.0885  &  0.0554  \\
Method 2  &  1.9204  &  0.0522  \\
Method 3  &  2.1806  &  0.0519 \\
\bottomrule
\end{tabular}
\label{table:errors_9layer}
\end{table}

The first thing to observe in Table~\ref{table:errors_9layer} is that the PSD errors are worse than with the 61-layer atmosphere across the board. The main reason for this appears to be the worse reconstructions given by method 1, as the relative errors in the PSD exponents are now \SI{1.93}{\percent} and \SI{1.36}{\percent} for 50 and 500 samples, respectively, whereas both errors were almost exactly \SI{1}{\percent} for the 61-layer atmosphere. The reduction in PSD error when moving from method 1 to methods 2 and 3 is also smaller than with the 61-layer atmosphere; this is mainly due to the few extra bumps which are visible in Fig.\ref{fig:psd_bumps_9layer}, as the simulated bumps are actually recovered better than in the 61-layer case. The source of these errors might be the simulated layer at \SI{140}{m}, as part of its turbulence may be interpreted as being located close to the ground, since the closest other option is also almost 90 meters away at \SI{228}{m}. It should also be noted that the error metrics we have chosen may not accurately represent the impact on reconstruction quality in atmospheric tomography, as a different PSD from the simulated one may be a better representation in the case where there is another layer close to the ground which follow standard von K\'arm\'an statistics. 

In terms of the residuals it also appears that we are doing worse than with the 61-layer atmosphere, as the numbers we get are larger relative to the true residual than they were before. This should of course be taken with a grain of salt, as a larger residual should naturally be expected when the simulation and reconstruction use different layer altitudes. However, we still see the same trend of improvement in the residuals as we move towards the more complex methods. In particular, both methods still provide better residuals for the data with 50 samples, which is somewhat impressive given the inherent model error from the difference in altitudes, and even with 500 samples method 3 gets close to the true residual and method 2 is even better. Overall, the results for the 9-layer atmosphere seem promising, even if they are not quite as good as with the 61-layer atmosphere.

\subsection{Robustness to parameter change}

We conclude this section by looking at the robustness of our methods to changes of $\beta_1$, $\beta_2$, $L_0$, and the number of discretization points for the ground layer. For each of these parameters, we have analyzed all six data sets presented in the previous section, using various methods to gauge how sensitive our results are to changing these parameters.

For $\beta_1$ and $\beta_2$, we determine the range of values where the reconstruction quality remains almost unchanged, and the range where we can at least locate the bumps in the PSD. We only consider changing the value of one parameter at a time. 

The parameter $\beta_1$ controls the $L^2$-penalty on the relative difference between the reconstructed and prior PSD. For all four data sets where methods 2 and 3 were used, the reconstruction quality is stable when $\beta_1$ is multiplied by a number between 0 and 3. Above this limit, the reconstruction tends towards the prior PSD, and the bumps become almost indistinguishable if $\beta_1$ is increased by a factor of 20. The main role of $\beta_1$ is thus to guarantee a stable solution, especially for method 3.

The parameter $\beta_2$ controls the $L^2$- or TV-penalty on the relative difference in derivatives between the reconstructed and prior PSD. The reconstruction quality for method 2 is stable when $\beta_2$ is multiplied or divided by $2$; for method 3, this factor drops to $1.5$. With method 2, the bumps can barely be recognized when $\beta_2$ is multiplied by 10 or divided by 5; for method 3, the corresponding factors are 3 and 4. This is certainly a point in favor of method 2, as there seems to be more leeway in choosing the value of $\beta_2$. Above these limits the solution gets too close to the prior, and below the lower limit the real bumps are obscured by noise. 

For the outer scale $L_0$, changing it from \SI{25}{\meter} to anything between \SI{13}{\meter} and \SI{10000}{\meter} leads to at most a \SI{1}{\percent} change in the coefficient and exponent given by method 1, as well as the relative turbulence profile error and the residual of all three methods. The PSD error for all methods is increased by at most $0.001$ for all $L_0$-values from \SIrange{22}{26}{\meter}; by contrast, the discretization error given by computing the PSD error for the simulated PSD is $0.0011$. For $L_0$-values from \SIrange{20}{30}{\meter} the increase in PSD error is at most $0.006$, but this is still relatively small as the PSD errors in the previous sections range from 0.017 to 0.067.

Finally, we consider the number of discretization points. We have experimented with reducing the number of points from 401 to 201 or 101 by only using every second or every fourth point, respectively. Changing the number of points has very little effect on method 1, as a smooth power law can be well represented even with a smaller number of points. Reducing the number of points to 101 changed the reconstructed exponent by at most \SI{0.1}{\percent}, and both the coefficient, residual and turbulence profile error changed by at most \SI{0.3}{\percent}. For the PSD error, the largest absolute increase was 0.0016; by contrast, the discretization error from representing the simulated PSD using 101 basis functions is 0.0047. Naturally, all of these errors are significantly smaller for 201 points, but for method 1 even as few as 101 points seems to be acceptable.

For methods 2 and 3, the regularization parameters $\beta_1$ and $\beta_2$ were multiplied by 2 for 201 points and by 4 for 101 points, as the regularization terms depend on the discretization. We found that with 201 points, the largest increases in residual, turbulence profile error and PSD error are \SIlist{2;0.16;0.65}{\percent} for method 2, and \SIlist{5.6;1;1}{\percent} for method 3. Method 3 thus clearly suffers more from reducing the number of points, although these errors are still quite reasonable. Reducing the number of points to 101, the largest increases in residual, turbulence profile error and PSD error are \SIlist{18;8.3;4.9}{\percent} for method 2, and \SIlist{17;1.8;3.5}{\percent} for method 3. Thus it seems that using less than 200 points for these two methods would be ill-advised.

\section{Conclusions}

We have generalized the SLODAR concept by introducing an isotropic turbulence layer with unknown statistics at the 
ground to the problem. We proposed an inverse problem of estimating the turbulence profile of a Kolmogorov model in the higher 
atmosphere and simultaneously the unknown power spectral density of the ground layer from the same data that SLODAR uses. 
It turns out the unknown ground layer statistics makes the problem quite ill-posed. 
It is therefore crucial to add an appropriate regularization or prior information that stabilizes the problem. 
We proposed three  regularization methods to compensate for the ill-posedness. The first method is in parametric form, 
where the unknown ground-layer power spectral density is represented by a finite-dimensional parameter space. 
The second method is non-parametric and is based on using Tikhonov regularization with a power-law favoring penalty term, 
and the third is also non-parametric but uses a mix of Tikhonov and total variation regularization.  

Our numerical results indicate that we can identify the turbulence profile and the ground layer statistics 
quite accurately, as long as the power spectral density does not deviate too much from a general von K\'arm\'an power law.
We also see a marked improvement in the turbulence profile reconstruction compared to the standard SLODAR method.

The ability to detect the power spectral density appears to be limited by the nature of SLODAR data, as the spatial frequency
detected by Shack--Hartmann wavefront sensors is limited by aliasing effects. This limitation can of course be mitigated in
the future as technological improvements allow for smaller WFS subapertures.

\section*{Acknowledgements}

TH was supported by the Academy of Finland via project 275177. 
JL acknowledges the support from the Jenny and Antti Wihuri foundation.
SK is supported by the Austrian Science Fund (FWF) project
P 30157-N31. 
TH and JL thank Gabriel Katul for fruitful discussion regarding atmospheric turbulence.
The authors wish to acknowledge CSC -- IT Center for Science, Finland, for computational resources.

\appendix
\section{A technical Lemma}\label{sec:appendix}
The function $g_h^\var(\vxi)$ in the kernel of the integral operator $K_h^\var$ 
allows for the following bound. 
\begin{lemma}\label{lemma:g}
There exists a constant $C > 0$ such that
\[ \int_{\partial B(0,r)} \abs{g_h^\var(\vxi)}^2 d\vxi \leq C(h) \]
for all $r > 1$ and $0 \leq h < H$.
\end{lemma}

\begin{proof}
Let us denote $c_h = \pi \eta(h) \ap$, and without loss of generality we let 
$\var = x$ and $\xi_\var = \xi_1$.
By elementary estimates for  trigonometric functions, we have
\begin{align*}
&\int_{\partial B(0,r)} \abs{g_h^x(\vxi)}^2 d\vxi = \int_{\partial B(0,r)} \abs{8 i \xi_1 \ap c_h \sin^2(c_h \xi_1) \sinc(c_h \xi_1 / \pi) \sinc (c_h \xi_2 / \pi)}^2 d\vxi \\
&\qquad \leq 64 c_h^2 \ap^2 \int_{\partial B(0,r)} \abs{\xi_1 \frac{\sin(c_h \xi_1)}{c_h \xi_1} 
\sinc (c_h \xi_2 / \pi)}^2 d\vxi \leq 64 \ap^2 \int_{\partial B(0,r)} \sinc^2(c_h \xi_2 / \pi) d\vxi.
\end{align*}
Observe that the function $\sinc^2(c_h \xi_2 / \pi)$ is constant in $\xi_1$ and even in the variable $\xi_2$,
and so the last term equals $4$ times the same expression but 
when integrated  over one quarter of the circle $\partial B(0,r)$. 
We denote the quarter of the circle where $\vxi$ is positive by $\partial B_+(0,r)$.
Now, we switch to polar coordinates and split the corresponding integral 
into two parts which we will estimate separately:
\begin{align*}
&\int_{\partial B_+(0,r)} \sinc^2(c_h \xi_2 / \pi) d\vxi = \int_0^{\pi/2} r \sinc^2(c_h r \sin(\theta) / \pi) d\theta \\
&\qquad = \int_0^{1/r} r \sinc^2(c_h r \sin(\theta) / \pi) d\theta + \int_{1/r}^{\pi/2} r \sinc^2(c_h r \sin(\theta) / \pi) 
d\theta =: I_1 + I_2.
\end{align*}
For $I_1$ we can use that  $\sinc^2(x) \leq 1$ to obtain
\[ I_1 = \int_0^{1/r} r \sinc^2(c_h r \sin(\theta) / \pi) d\theta \leq \int_0^{1/r} r d\theta = 1. \]
For $I_2$, we have 
\[ I_2 = \int_{1/r}^{\pi/2} r \sinc^2(c_h r \sin(\theta) / \pi) d\theta = \int_{1/r}^{\pi/2} r \frac{\sin^2(c_h r \sin(\theta))}{c_h^2 r^2 \sin^2(\theta)} d\theta \leq \frac{1}{c_h^2 r} \int_{1/r}^{\pi/2} \frac{1}{\sin^2(\theta)} d\theta. \]
The integral function of $1/\sin^2(\theta)$ is $-\cos(\theta)/\sin(\theta)$, which vanishes at $\pi/2$. Thus,
\[ I_2 \leq \frac{1}{c_h^2 r} \int_{1/r}^{\pi/2} \frac{1}{\sin^2(\theta)} d\theta = \frac{\cos(1/r)}{c_h^2 r \sin(1/r)} \leq \frac{1}{c_h^2 r \sin(1/r)}. \]
Finally, the estimate $\sin(1/r) \geq \sin(1)/r$ for all $r \geq 1$, yields 
 $I_2 \leq \tfrac{1}{c_h^2 \sin(1)}$, and
consequently
\[ \int_{\partial B(0,r)} \abs{g_h^\var(\vxi)}^2 d\vxi \leq 
256 \ap^2 \int_{\partial B_+(0,r)} \sinc^2(c_h \xi_2 / \pi) d\vxi \leq 
256 \left(\ap^2 + \frac{1}{\pi^2 \eta(h)^2 \sin(1)}\right). \] 
By symmetry, the same estimate holds for the case $\var = y$. Note that since $\eta(h) = 1 - \tfrac{h}{H}$, this limit is well-defined for all $0 \leq h < H$, and tends to infinity as $h \to H$.
\end{proof}

\bibliographystyle{plain} 
\bibliography{main}

\end{document}